\documentclass[twocolumn,english,superscriptaddress,prx,floatfix]{revtex4-2}

\usepackage{graphicx}
\usepackage{amsmath}
\usepackage{amsthm}
\usepackage{bm}
\usepackage{layout}
\usepackage{float}
\usepackage{amsfonts}
\usepackage{amssymb}%
\usepackage{color}
\usepackage{soul}
\usepackage{mathtools}
\usepackage{varwidth}
\usepackage[colorlinks=true,citecolor=blue]{hyperref}
\theoremstyle{definition}
\newtheorem{definition}{Definition}

\usepackage[capitalise,compress]{cleveref}

\newcommand{\oket}[1]{\left | #1 \right)}

\newcommand{\Tr}{\mathrm{Tr}}

\newcommand{\abs}[1]{\left | #1 \right|}
\renewcommand{\epsilon}{\varepsilon}
\renewcommand{\O}[1]{ O\left(#1\right)}
\newcommand{\norm}[1]{\left\|#1\right\|}

\newcounter{para}
\newcommand{\dist}{\textrm{dist}}

\makeatletter
\newcommand*\bigcdot{\mathpalette\bigcdot@{.5}}
\newcommand*\bigcdot@[2]{\mathbin{\vcenter{\hbox{\scalebox{#2}{$\m@th#1\bullet$}}}}}

\makeatother

\newcommand{\tr}{\text{Tr}}

\usepackage{braket}

\usepackage{array}
\newcolumntype{L}{>{$}l<{$}} 
\newcolumntype{C}{>{$}c<{$}} 
\newcolumntype{R}{>{$}r<{$}} 

\usepackage{xcolor}

\usepackage{mathtools}

\usepackage{xr}
\makeatletter
\newcommand*{\addFileDependency}[1]{
  \typeout{(#1)}
  \@addtofilelist{#1}
  \IfFileExists{#1}{}{\typeout{No file #1.}}
}
\makeatother

\usepackage{tikz}

\renewcommand{\L}{\mathcal{L}}
\newcommand{\Py}{\mathbb{P}_Y}
\newcommand{\length}{\text{length}}

\usepackage{graphicx}
\usepackage{bm}
\usepackage{amsmath}
\usepackage{amsfonts,dsfont}
\usepackage{subfigure}
\usepackage{xr} 
\usepackage{times}
\usepackage{color}
\usepackage[colorlinks=true,citecolor=blue]{hyperref}
\usepackage{physics}
\usepackage{algpseudocode}
\usepackage{algorithm}
\usepackage{mathtools}
\usepackage{amsthm}
\usepackage{qcircuit}
\usepackage{comment}

\newcommand{\eps}{\ensuremath{\varepsilon}}

\newcommand{\ra}{\ensuremath{\rightarrow}}

\newcommand{\ds}{\ensuremath{\displaystyle}}

\newcommand{\al}{\ensuremath{\alpha}}
\newcommand{\lam}{\ensuremath{\lambda}}
\newcommand{\Lam}{\ensuremath{\Lambda}}
\renewcommand{\norm}[1]{\left\|#1\right\|}

\newcommand{\mL}{\L}

\newcommand{\mQ}{\mathcal{Q}}

\DeclarePairedDelimiter\parentheses{\lparen}{\rparen}
\renewcommand{\O}[1]{\mathcal{O}\parentheses*{#1}}
\newcommand{\W}[1]{\Omega\parentheses*{#1}}
\newcommand{\Th}[1]{\Theta\parentheses*{#1}}

\renewcommand{\L}{\mathcal{L}}

\newcommand{\td}{\tilde}

\DeclarePairedDelimiter\ceil{\lceil}{\rceil}
\DeclarePairedDelimiter\floor{\lfloor}{\rfloor}

\newcommand{\etal}{\emph{et al.}}

\newtheorem{theorem}{Theorem}
\newtheorem{lemma}{Lemma}

\usepackage[capitalise,compress]{cleveref}
\crefname{section}{Sec.}{Secs.}
\crefname{section}{Section}{Sections}
\crefrangelabelformat{equation}{\textup{(#3#1#4)}--\textup{(#5#2#6)}}

\begin{document}

\title{Clustering of steady-state correlations in open systems with long-range interactions}
\author{Andrew Y. Guo}
\email[Corresponding author: ]{guoa@umd.edu}
\affiliation{Joint Center for Quantum Information and Computer Science, NIST/University of Maryland, College Park, MD 20742, USA}
\affiliation{Joint Quantum Institute, NIST/University of Maryland, College Park, MD 20742, USA}
\author{Simon Lieu}
\thanks{The first two authors contributed equally}
\affiliation{Joint Center for Quantum Information and Computer Science, NIST/University of Maryland, College Park, MD 20742, USA}
\affiliation{Joint Quantum Institute, NIST/University of Maryland, College Park, MD 20742, USA}
\thanks{These two authors contributed equally}
\author{Minh C. Tran}
\affiliation{Joint Center for Quantum Information and Computer Science, NIST/University of Maryland, College Park, MD 20742, USA}
\affiliation{Joint Quantum Institute, NIST/University of Maryland, College Park, MD 20742, USA}
\affiliation{Center for Theoretical Physics, Massachusetts Institute of Technology, Cambridge, Massachusetts 02139, USA}
 \affiliation{Department of Physics, Harvard University, Cambridge, Massachusetts 02138, USA}
\author{Alexey V. Gorshkov}
\affiliation{Joint Center for Quantum Information and Computer Science, NIST/University of Maryland, College Park, MD 20742, USA}
\affiliation{Joint Quantum Institute, NIST/University of Maryland, College Park, MD 20742, USA}

\date{\today}

\begin{abstract}
Lieb-Robinson bounds are powerful analytical tools for constraining the dynamic and static properties of non-relativistic quantum systems.
Recently, a complete picture for closed systems that evolve unitarily in time has been achieved.
In experimental systems, however, interactions with the environment cannot generally be ignored, and the extension of Lieb-Robinson bounds to dissipative systems which evolve non-unitarily in time remains an open challenge.
In this work, we prove two Lieb-Robinson bounds that constrain the dynamics of open quantum systems with long-range interactions that decay as a power-law in the distance between particles.
Using a combination of these Lieb-Robinson bounds and mixing bounds which arise from ``reversibility''---naturally satisfied for thermal environments---we prove the clustering of correlations in the steady states of open quantum systems with long-range interactions.
Our work provides an initial step towards constraining the steady-state entanglement structure for a broad class of experimental platforms, and we highlight several open  directions regarding the application of Lieb-Robinson bounds to dissipative systems.
\end{abstract}

\maketitle

\section{Introduction}

While the  speed of information transfer  is always bounded by the speed of light, many quantum platforms operate in a non-relativistic regime where typical velocities are far below this threshold.
Nevertheless, the  Schr\"odinger equation admits fundamental limits to the rate  at which correlations  can spread throughout the system.
Such bounds are known as Lieb-Robinson bounds and are connected to a diverse array of phenomena, including the decay of correlations in the ground state \cite{Hastings2006}, generation of topological order \cite{Bravyi2006, Bravyi2010}, efficiency of classical/quantum simulation \cite{Osborne2006,Tran2019a}, hardness of bosonic sampling tasks \cite{Deshpande2018}, heating rates in periodically driven Floquet systems \cite{Abanin2015,Tran2019b}, and signatures of quantum chaos \cite{Lashkari2013,Guo2019}.

To date, most  formulations of Lieb-Robinson bounds apply to closed systems that evolve  via a unitary time-evolution operator.
In such systems, recent advances have proved tight information-transfer bounds for interaction ranges that span the whole spectrum from local \cite{ChenLucas2021graphtheory,WangHazzard2020} to highly non-local regimes \cite{Tran2019a,Chen2019,kuwaharaStrictlyLinearLight2020,Tran2021b}, and have been saturated via explicit state-transfer protocols \cite{Eldredge2017,Tran2020hierarchylinearlightcones,Tran2021}.
While a complete picture for quantum information transfer has emerged for closed quantum systems,  the analogous question for systems that evolve \textit{non-unitarily} in time remains less well understood.
For a broad range of quantum platforms (including noisy quantum simulators), interactions with a  larger environment are unavoidable and must be taken into account to accurately describe dynamics.
While progress in this direction has been made \cite{poulin2010, descamps2013, cubitt2015, Kastoryano2013, Sweke2019}, the question of how the fundamental rate of information transfer differs for systems that interact with some larger environment remains unanswered.

Indeed, the notion of a  Lieb-Robinson bound in an open system may seem \textit{a priori}  surprising from the point of view of quantum trajectories \cite{Knight1998}. In this picture,  in a time-step $dt$ the system's wavefunction either evolves via a non-unitary evolution operator, or a quantum jump discontinuously  alters the state.
A specific  trajectory belonging to a spatially-local Hamiltonian with local dissipation can transfer information faster than  the limit set by the Hamiltonian's Lieb-Robinson bound \cite{Ashida2018}. Intuitively, this is because conditioning on measurements is an inherently nonlocal process.  As an extreme example, it is possible to create a highly-entangled (GHZ) state from a product state in a time $t = \O{1}$ using only locally entangling gates  and measurements, for a specific outcome of the measurements \cite{Pham2013}.  This would violate the Lieb-Robinson bound for local systems, which gives $t = \W{r}$ for distance $r$ \cite{bigO}. After averaging over trajectories, the state of the system can be represented via a density matrix $\rho$ which evolves via a master equation: $ d \rho / dt =  \mL (\rho) $. Subsequently, the notion of a Lieb-Robinson bound is properly restored upon averaging over trajectories.

In this work, we make progress on the question of the fundamental rates of information propagation in open systems by proving a broad class of Lieb-Robinson bounds for systems with long-range interactions---specifically those that decay as a power-law $1/r^\alpha$ in the distance  $r$ between particles, for some $\alpha > 0$.
Such power-law-decaying interactions feature in experimental platforms relevant to quantum computation and simulation, such as Rydberg atoms~\cite{Saffman2010}, trapped ion crystals~\cite{Britton2012,Monroe2021}, polar molecules \cite{Yan2013}, and nitrogen-vacancy color centers in diamond \cite{Yao2012}.
In all of these platforms, interactions with a larger environment cannot be neglected, and a Markovian description of system dynamics is often justified.  In such systems, improved understanding of the fundamental rates of information transfer has spurred the development of optimal protocols for quantum information processing and state transfer \cite{Eldredge2017,Tran2021}.

In addition to bounding dynamics of open long-range systems, we use these Lieb-Robinson bounds to constrain the entanglement structure of the corresponding steady states.
For closed systems, Lieb-Robinson bounds have played an important role in proving rigorous statements on the decay of correlations in gapped ground states \cite{Hastings2006}. This justifies the use of finite-dimensional matrix-product-state representations of the ground state in one-dimensional systems with local interactions \cite{Hastings2007}.
In this work, we prove the clustering of correlations in the steady states of open power-law systems, which may serve as a first step towards establishing an area-law scaling of entanglement for these systems, similar to what was done in Ref.~\cite{Gong2017} for the closed case.

This paper is organized as follows:  in \cref{sec:open-LR}, we summarize the existing Lieb-Robinson bounds for open long-range systems and present two new bounds that are tighter for particular regimes of the power-law exponent $\alpha$.
The first yields a polynomial light cone for $\al > 2d$, using a technique pioneered in Ref.~\cite{Tran2019a}.
The second gives a linear light cone for $\al > 3$ in 1D, using the method from Ref.~\cite{Chen2019}.
In \cref{sec:clustering-of-correlations}, we also prove the clustering of correlations in the steady states of open long-range systems.
Specifically, we provide bounds on the extent of the covariance correlations and mutual information under certain assumptions on the Liouvillian mixing times.
We also prove a stability theorem for the stationary state under local Liouvillian perturbations, generalizing the results of Ref.~\cite{Kastoryano2013}.

\section{Lieb-Robinson bounds for open long-range systems}
\label{sec:open-LR}
In this section, we review the results of the previous best-known Lieb-Robinson bounds for open long-range systems and state two new Lieb-Robinson bounds.

As a general set-up, we consider evolution by a long-range Liouvillian $\L(t)$ that acts on a lattice $\Lambda$ consisting of finite-level systems at each site.
We denote by $\mathcal{H}$ the finite-dimensional Hilbert space representing all possible states of the system and by $\mathcal {B(H)}$ the space of all operators on $\mathcal{H}$.
For an operator $O \in \mathcal {B(H)}$, we will be interested in how its expectation value changes as a function of time: $\langle O(t) \rangle = \tr[ O(t) \rho] = \tr[ O \rho(t)]$, where $\rho$ is the initial state of the system, which evolves (in the Schr\"odinger  picture) via $\rho(t) = e^{\L t} \rho$. For these purposes,
the time-evolution of $O$ can be expressed as $O(t) = e^{\L^\dag t} O$, where $\L^\dag$ is the adjoint Lindblad superoperator, defined as
\begin{align}
	\L^\dag O = +i[H,O] + \sum_i \left[L_i^\dag O L_i - \frac12 \{L_i^\dag L_i,O\}\right],
\end{align}
where $H$ is the Hamiltonian and $L_i$ are Lindblad operators (also referred to as ``jump'' operators) \cite{Breuer2010}.  We emphasize that $O(t)$ is \textit{not} equivalent to the Heisenberg-Langevin time evolution for the operator $O$. For example, if the system has a unique steady state, all operators $O(t)$ will be proportional to the identity at long times: $\lim_{t \rightarrow \infty} O(t) \sim \mathbb{I}$. Thus two operators that do not commute at $t=0$ will start to commute at long times.

We will state the Lieb-Robinson-type bounds in this paper in terms of time-independent Liouvillians. However, we note that the proofs can be generalized with minor modifications to the case of time-dependent Liouvillians---i.e. those for which both $H$ and $L_i$ are allowed to vary in time.

To impose the long-range condition on $\L$, we decompose it into $\L = \sum_{Z\subset \Lambda} \L_Z$, where for any pair of sites $i,j$, we have the condition
\begin{align}
  \label{eq:L-norm-bound}
    \sum_{Z\ni i,j}\|\L_{Z}\|_{\infty} \coloneqq \sup_{O\in\mathcal{B(H)}} \frac{\|\L_{Z}O\|}{\|O\|}\le \frac{1}{\text{dist}(i,j)^\al},
\end{align}
where $\|\cdot\|$ denotes the standard operator norm, or $\infty$-norm, and  $\|\cdot\|_{\infty}$ denotes the superoperator, or ``$\infty \ra \infty$'' norm (referred to as such because the second term in \cref{eq:L-norm-bound} uses the operator $\infty$-norm in both the numerator and the denominator).
Here $\text{dist}(i,j)$ is the distance between $i$ and $j$, and $\al$ is the positive real parameter that controls the long-ranged nature of the interaction.

\subsection{Prior work on open-system Lieb-Robinson bounds}
In Ref.~\cite{Sweke2019}, Sweke \etal~generalized the Lieb-Robinson bound in Ref.~\cite{Hastings2006} for $\al > d$ to open systems.
Letting $A\in \mathcal B(X)$ be an operator supported on $X$, $K_Y\in \mathbb{L}_Y$ be a Liouvillian supported on $Y$, and $e^{\L^\dag t}$ be the evolution under the adjoint Liouvillian superoperator.
The corresponding superoperator bound is:
\begin{align}
  \norm{  K_Y(e^{\L^\dag t}A)} \leq C \|K_Y\|_{\infty} \norm A  \abs{X}\abs{Y}
    \left(
    \frac{e^{vt}-1}{r^{\alpha}}
    \right)
    ,\label{eq:LR-HK-open}
\end{align}
where $r\coloneqq d(X,Y)$, and $C$ and $v$ are $\O1$ constants.
In the closed-system picture, the conventional Lieb-Robinson-type bound can be recovered by choosing $K_Y$ such that $K_Y (O_X) = i[O_X,O_Y]$ and replacing $\|K_Y\|_{\infty}$ with $2\|O_Y\|$.

For this conventional bound, the velocity scales as $v\propto 2^\al $, which diverges in the limit $\al \ra \infty$.
To recover the Lieb-Robinson bound for short-range interacting systems, an improved bound is required that uses a slight modification of the technique from Ref.~\cite{Sweke2019}:
\begin{align}
  \begin{split}
  \norm{  K_Y(e^{\L^\dag t}A)} \leq C \|K_Y\|_{\infty} \norm A  \abs{X}\abs{Y}
    \biggl(&
    \frac{e^{\td vt}}{[(1-\mu)r]^{\alpha}}\\
    &+e^{\td vt-\mu r}
    \biggr),
  \end{split}\label{eq:LR-ZX-open}
\end{align}
where $\mu\in (0,1)$ and $\td v$ are constants, and $\td v$ is independent of $\al$.
The closed-system version of this bound was first proven in Ref.~\cite{Gong2014} for two-body interactions and later generalized to $k$-body interactions in Ref.~\cite{Tran2019b}.
In Ref.~\cite{Sweke2019}, Sweke \etal~also prove Lieb-Robinson-type bounds for $\al \le d$.
For this regime of $\al$, one needs to restrict to a finite-sized lattice, due to the energy being (in general) non-additive for subsystems \cite{Dauxois}.
Denoting the system size of the lattice by $N\coloneqq \abs{\Lambda}$,
the combined strength $J$ of the terms acting on a single site  scales as $J = \Th{N^{1-\al/d}}$ for $\al < d$ and $J = \Th{\log{N}}$ for $\al = d$ \cite{bigO,Guo2019}.
The bound then becomes:
\begin{align}
\label{eq:LR-ZX-open-small-alpha}
    \norm{  K_Y(e^{\L^\dag t}A)} \leq C \|K_Y\|_{\infty} \norm A  \abs{X}\abs{Y}
  \left(
  \frac{e^{J t}-1}{J r^{\alpha}}
  \right).
\end{align}
The effective Lieb-Robinson velocity in this case diverges in the thermodynamic limit, but is finite for all finite $N$.

\subsection{Power-law light-cone bound for $\al > 2d$}
\label{sec:power-law-bound}

We prove a Lieb-Robinson bound  for $\al >2d$ using the truncation-of-unitaries-approach presented by Tran \etal~\cite{Tran2019b}.
The technique takes as input the existing open-systems Lieb-Robinson bound in \cref{eq:LR-ZX-open} and bootstraps it to obtain a tighter bound:
\begin{align}
  \norm{  K_Y(e^{\L^\dag t}A)}
\leq C \|K_Y\|_{\infty}\norm A \frac{t^{\alpha-d}}{r^{\alpha-2d}}.
    \label{eq:LR-Minh-constX}
\end{align}

This bound yields the power-law light-cone contour $t = r^{\frac{\al-2d}{\al-d}}$.
The proof of this bound involves approximating the time evolution of the operators by a sequence of operators that  span successively larger and larger subsets of the lattice, and bounding the error of each successive approximation  by the existing Lieb-Robinson bound.
We provide the full details of the derivation in \cref{sec:minh-bound-proof}.

\subsection{Linear light-cone bound for $d=1$, $\al > 3$}
\label{sec:chen-lucas-bound}
Finally, we prove a bound with a linear light cone for open-long-range systems with $\al > 3$ in $d=1$ dimensions based on the techniques developed in Ref.~\cite{Chen2019}.
In the process, we tighten the tail of the Lieb-Robinson bound given in that work from $1/r$ to approximately $1/r^{\alpha-2}$.
The authors of Ref.~\cite{Chen2019} proved the following bound for the closed-system dynamics of Hamiltonian $H = \sum_{ij} H_{ij}$ consisting of two-body terms:
\begin{align}
\label{eq:chen-lucas-closed-bound}
	\norm{ [e^{iHt} A e^{-iHt},B]} \leq C \norm A \norm{B} \frac{t}{r},
\end{align}
where $B \in \mathcal{B}(Y)$ is an operator supported on $Y$.
Likewise assuming a two-body Liouvillian $\L = \sum_{ij} \L_{ij}$, we obtain the following open-systems bound:
\begin{align}
\label{eq:chen-lucas-open-bound}
	\norm{  K_Y(e^{\L^\dagger t} A)} \leq C \norm{K_Y}_{\infty} \norm A \frac{t}{r^{\alpha-2-o(1)}},
\end{align}
where the $o(1)$ denotes some constant that can be made arbitrarily small.
The result yields a linear light cone $t \gtrsim r$ for all $\al > 3$.
The proof roughly proceeds by expanding out the evolution operator $e^{\L^\dagger t}$ into a series of products of Liouvillian terms $\L_{ij}$.
For each term in the series, we select out a subsequence of terms that move the operator forward (i.e.~towards $Y$) and integrate out the other terms.
By only taking into account the contributions from the terms in the subsequences,  we are able to obtain a tighter Lieb-Robinson bound.
We provide the mathematical details of the proof in \cref{sec:chen-lucas-bound-proof}.

\section{Bounds on correlations in the steady states of open long-range systems}
\label{sec:clustering-of-correlations}
In this section, we prove the clustering of correlations in the steady states of open long-range systems.
We first state a lemma that describes how to use a modified version of the Lieb-Robinson bounds stated in the previous section to bound how far operators can spread under evolution by the (adjoint) Liouvillian $\L^\dagger$.
Specifically, we give a bound on the error of approximating the time-evolution of an operator $A$ supported on a site $X\in\Lam$ by a truncated adjoint Liouvillian that only acts on ball of radius $r$ centered on a site $X\in\Lam$ (see \cref{fig:LRbound_truncated}).

\begin{figure}
  \includegraphics[scale=0.875]{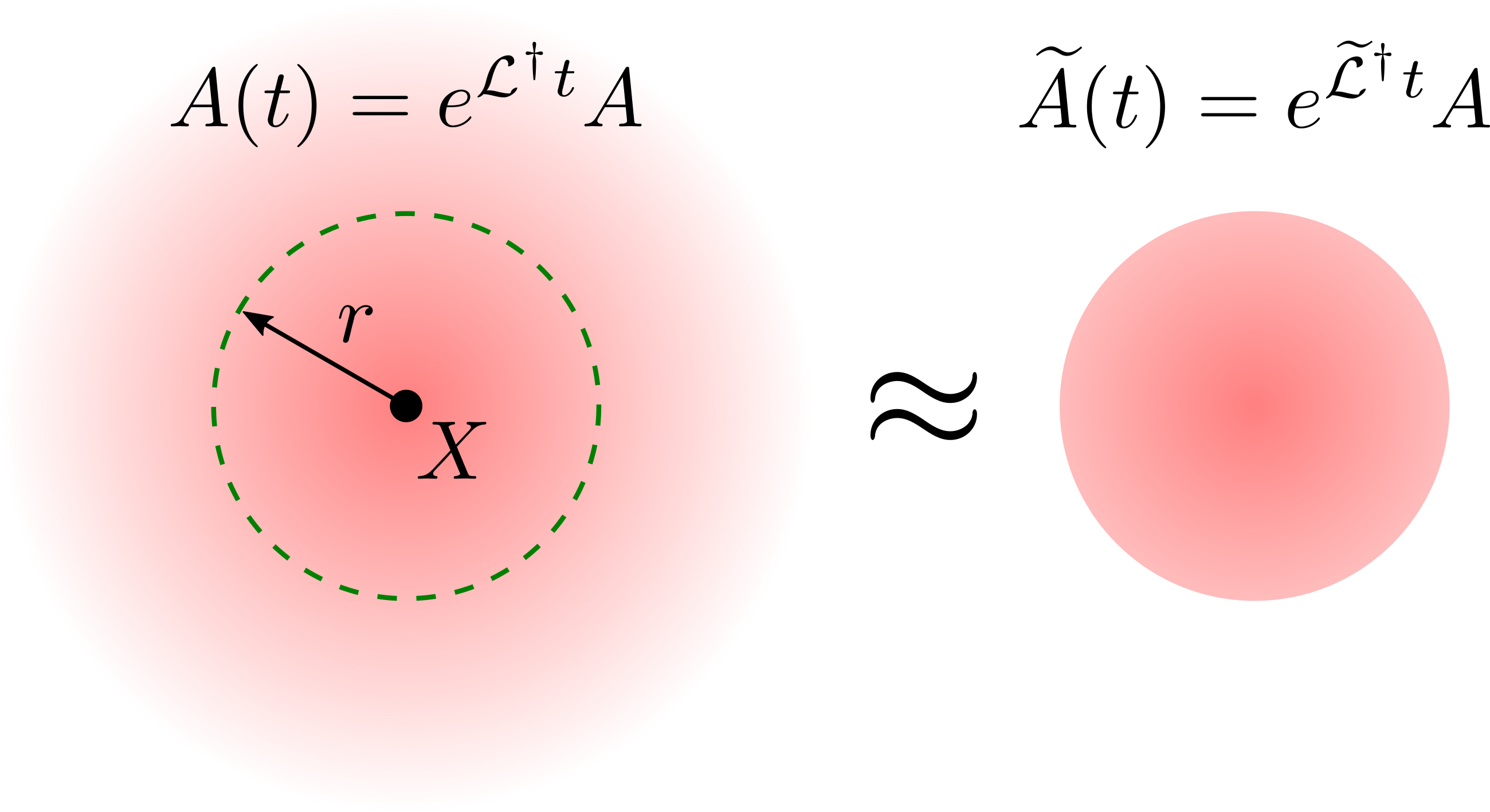}
  \caption{The evolution of an operator $A$ initially supported on $X$ by an adjoint Liouvillian $\L^\dag$ can be approximated by the same operator evolved by the truncated version of the Liouvillian, $\td \L^\dag$, supported on a ball of radius $r$ around $X$, up to an error given by $C(r,t)$.}
  \label{fig:LRbound_truncated}
\end{figure}

\begin{lemma}
[Bounds on the error incurred by approximating of time-evolved operators by local ones]
\label{lemma:LRbound_truncated}
    Let $A$ be an operator initially supported on a site $X\in\Lam$ and let $\td \L$ be the restriction of the long-range Liouvillian $\L$ to the ball of radius $r$ centered on $X$.
    Let $\tilde A(t)$ be the time-evolved version of $A$ under $\td \L^\dag$.
    Then the error in the approximation of $A(t)$ by $\tilde A(t)$ is bounded by
    \begin{align}
    \label{eq:LR-bound-approx}
        \|A(t) - \tilde A(t) \| \le K \|A\|\,\mathcal C(r,t),
    \end{align}
    where $K$ is some constant, and $\mathcal C(r,t)$ is a modified version of the standard Lieb-Robinson bound adapted to the problem of locally approximating time-evolved operators.
    In the large-$r$ limit, the tightest-known bounds for open systems with long-range interactions with $\alpha > d$ scale asymptotically as
    \begin{equation}
      \label{eq:LR-bound-cc}
      \ds
      \mathcal C(r,t) \propto \begin{cases} \ds
            e^{vt}/r^{\al-d},& \al > d,
        \\
        \ds t^{\alpha-d+1}/r^{\alpha-3d}
        , & \al > 3d,
        \\ \ds t^2/r^{\alpha-3},& \al > 3, d=1.
      \end{cases}
    \end{equation}
    For $\al \le d$, the bounds also depend on the system size of the lattice $N\coloneqq|\Lam|$.
    When $r \propto N^{1/d}$, the bounds scale as follows:
    \begin{equation}
      \label{eq:LR-bound-cc-small-alpha}
      \mathcal C(N,t) \propto \begin{cases}
        \ds\frac{e^{\Th{N^{1-\al/d}}t}-1}{\Th{N^{1-\al/d}}},& \al < d,
        \\
        \ds\frac{e^{\Th{\log(N)}t}-1}{\Th{\log(N)}},& \al= d.
      \end{cases}
    \end{equation}
    This concludes the statement of \cref{lemma:LRbound_truncated}.
\end{lemma}
The proof of \cref{lemma:LRbound_truncated} follows straightforwardly from the open-system Lieb-Robinson bounds detailed in \cref{sec:open-LR}.
In particular, the three lines of \cref{eq:LR-bound-cc} follow from \cref{eq:LR-ZX-open}, \cref{eq:LR-Minh-constX}, and \cref{eq:chen-lucas-open-bound}, respectively,
while \cref{eq:LR-bound-cc-small-alpha} comes from \cref{eq:LR-ZX-open-small-alpha}.
In \cref{app:sim-local-obs}, we provide the details of the derivation of the bounds in \cref{lemma:LRbound_truncated}.
We now proceed to derive the bounds on clustering of correlations in the steady states of gapped, reversible Liouvillians.

\subsection{Bound on covariance correlations}
In this first section, we show how open-system Lieb-Robinson bounds constrain the correlations in the steady state of a Liouvillian $\L$ with dissipative gap $\lambda$.

The dissipative gap $\lambda >0$
is defined as the magnitude of the least-negative non-zero real part of an eigenvalue of $\L$.
(Throughout this work, we shall also assume that the Liouvillian is \textit{primitive}, i.e.~it has a unique steady state such that $\L$ has one eigenvalue of zero.)
In addition to the Lieb-Robinson bounds,  we will also appeal to certain ``mixing bounds'' which  describe how fast  arbitrary initial states (or various correlation functions) converge to the steady state.

For the mixing bounds, we need to impose ``reversibility'' on the Liouvillian. We say that a Liouvillian $\L$ is $s$-reversible if there exists some operator $\sigma$ such that $\Gamma_s \L^\dag = \L \Gamma_s$ is satisfied;  the superoperator $\Gamma_s$ is defined via $\Gamma_s(f) = ( \sigma^s f \sigma^{1-s} + \sigma^{1-s} f \sigma^s)/2$ and $s \in[0,1]$. For $s$-reversible Liouvillians, it is easy to see that $\sigma$ is the steady state of $\L$ (since $\L^\dag (\mathbb{I}) = 0$).
A sufficient condition for a Liouvillian to satisfy $s$-reversibility (for all $s$) is if the dissipators  $L_i$ satisfy a detailed-balance condition (and the Hamiltonian is zero, $H=0$), which is naturally obeyed for systems coupled to a thermal bath \cite{Kastoryano2013}. (More explicitly, the detailed-balance condition is satisfied if dissipators come in energy raising/lowering pairs with respect to some effective Hamiltonian $\bar{H}$---for example, if $[\bar{H}, L_{\pm}] = \pm \omega L_{\pm}$ and $|L_-| / |L_+| = \exp(2 \beta \omega)$ where $\beta^{-1}$ is an effective temperature.)

Returning to the topic of correlations, we let $\rho$ be a quantum state defined on the lattice $\Lambda$.
We are interested in the \emph{covariance correlation} between non-overlapping $X,Y \in \Lambda$:
\begin{align}
\label{eq:covcorrdef}
  T_\rho (X:Y) := \text{sup}_{\norm{f} = \norm{g} = 1}|\text{Tr}[(f \otimes g)(\rho_{XY} - \rho_X \otimes \rho_Y) ] |,
\end{align}
where $f$ and $g$ are Hermitian operators with $f$ supported on region $X$ and $g$ supported on region $Y$, and where $\rho_{X}$ is the reduced density matrix constructed from $\rho$ by tracing over the complement of $X$. Our goal is to bound this correlation function in terms of $\lam$ and the distance between $X$ and $Y$.

We follow Theorem 9 in Ref.~\cite{Kastoryano2013}.
Let $\sigma$ be the steady state of the Liouvillian.
From the right-hand side of \cref{eq:covcorrdef}, we define
\begin{equation}
\text{Cov}_\sigma(  f, g) \coloneqq \frac{1}{2} \text{Tr}[(f g + g f) \sigma] -\text{Tr}[f \sigma] \text{Tr}[g \sigma],
\end{equation}
which is equivalent to the term inside the sup (because $f$ and $g$ commute).
Now we use the bound (which follows directly from the triangle inequality)
\begin{equation}
\label{eq:covcorr}
  |\text{Cov}_\sigma(f, g)| \leq |\text{Cov}_\sigma (f_t, g_t)| + |\text{Cov}_\sigma(f, g)  - \text{Cov}_\sigma(f_t, g_t)|.
\end{equation}
Here $f_t$ and $g_t$ are the time-evolved versions of $f$ and $g$ under $\L^\dag$.
This step allows us to relate a \textit{static} covariance to  time-dependent quantities;  we will use dynamical bounds to  constrain the form of the latter, then pick an optimal time which maximally bounds the static covariance.

The first term on the right is constrained by the variance bound for $s$-reversible, primitive Liouvillians (see Appendix \ref{sec:var-bound})
\begin{equation} \label{eq:cov-bound}
|\text{Cov}_\sigma (f_t, g_t)| \leq 4 \norm{f} \norm{g}  e^{-2 \lambda t},
\end{equation}
where $\lambda$ is the dissipative gap of $\L$.
Intuitively, this relationship can be understood as follows: the operators $f_t, g_t$  both evolve (in time) toward an operator  that is proportional to the identity, so the covariance between them will eventually tend to zero as a function of time. The rate at which this occurs is set by the dissipative gap of the system.

To bound the second term, we use the relation $\Tr[\sigma f_t]=\Tr[\sigma f]$, which holds for all observables $f$.
This gives:
\begin{align}
& |\text{Cov}_\sigma(f, g)  - \text{Cov}_\sigma(f_t, g_t)| \\
& = \frac12 (|\Tr[(fg-f_tg_t)\sigma] + \Tr[(gf-g_tf_t)\sigma]|) \\
& = \frac12 (|\Tr[((fg)_t-f_tg_t)\sigma] + \Tr[((gf)_t-g_tf_t)\sigma]|) \\
& \leq \frac12 (\|(fg)_t-f_tg_t\|+\|(gf)_t -g_tf_t\|) \\
& \le K \norm{f}  \norm{g} \mathcal{C}(r, t), \label{eq:secondterm}
\end{align}
where $r \coloneqq d(X,Y)$.  We obtain the inequality in the final line using the open-system Lieb-Robinson bounds $\mathcal C(r,t)$ given in \cref{lemma:LRbound_truncated}.
Specifically, we use the following Lemma, which is itself a restatement of Corollary 7 in Ref.~\cite{Kastoryano2013}:
\begin{lemma}
[Time-evolution of spatially separated observables]
  \label{lemma:connectedcorrs}
    Take two operators $A$ and $B$ supported on $X,Y \in \Lam$  respectively such that $r\coloneqq d(X,Y)$, and let $A(t)=e^{\L^\dag t}A$ and $B(t) = e^{\L^\dag t}B$ be their time-evolution under the adjoint Liouvillian $\L^\dag$.
    We also define $(AB)(t) = e^{\L^\dag t}(AB)$.
    Then the following bound holds:
    \begin{align}
    \label{eq:ops-evolving-together}
        \|(AB)(t) - A(t)B(t)\| \le K\|A\|\|B\|\mathcal C(r,t),
    \end{align}
    where $\mathcal C(r,t)$ is given by \cref{lemma:LRbound_truncated} and $K$ is some constant that depends on lattice parameters.
\end{lemma}
\cref{lemma:connectedcorrs} bounds the difference between operators that evolve together in the Heisenberg picture as opposed to evolving separately.
Again we emphasize that $A(t)$, $B(t)$, and $(AB)(t)$ are \textit{not} equivalent to Heisenberg-Langevin evolution, a fact that is at the core of this bound.
We defer the short proof of \cref{lemma:connectedcorrs} to \cref{app:connectedcorrs} and move on to proving the bound on the covariance correlations.

\begin{theorem}[Bounds on steady-state covariance correlations]
\label{theorem:covariancebound}
Consider Hermitian operators $f,g$ which are supported on two non-overlapping subsets $X$ and $Y$ of the $d$-dimensional cubic lattice $\Lambda$, and let $\L$ be an $s$-reversible Liouvillian with stationary state $\sigma$ and dissipative gap $\lambda$ that satisfies the conditions in \cref{eq:L-norm-bound}.
Then there exists a constant $c > 0$ which only depends on $\lambda, v$  such that
\begin{equation}
\label{eq:bounds-covar-correl}
T_\sigma (X:Y) \leq \begin{cases}
     \ds  c \left( r^{\alpha - d} \right)^{\frac{-2 \lambda}{ v + 2\lambda}},& \al > d,
    \\ \ds c \frac{\log(r)^{\al-d+1}}{r^{\al-3d}},& \al > 3d,
    \\ \ds c\frac{\log(r)^2}{r^{\al-3}},& \al> 3, d=1.
    \end{cases}
\end{equation}
\end{theorem}

\begin{proof}
From our previous analysis [see Eqs.~(\ref{eq:cov-bound}) and (\ref{eq:secondterm})] on the covariance correlation in \cref{eq:covcorr}, we have
\begin{equation}
|\text{Cov}_\sigma (f, g)| \leq 4 \norm{f} \norm{g}  \left(e^{-2 \lambda t} + \frac K4 \mathcal{C}(r, t) \right).
\end{equation}
To obtain the tightest bound, we  minimize with respect to $t$ the function
\begin{equation}
\label{eq:cov_corr_minimization}
h(t) =   e^{-\lambda' t} + K' \mathcal{C}(r, t),
\end{equation}
where $\lambda' = 2\lambda, K'= K/4$.

We will perform this minimization exactly for the first case in \cref{eq:bounds-covar-correl}, for which $\mathcal{C}(r,t)$ is given by the first line of \cref{eq:LR-bound-cc}; for the other cases, we instead use an approximation to the optimal ansatz, which allows us to obtain an analytical expression for the bound.
Setting $dh/dt=0$ in \cref{eq:cov_corr_minimization} leads to a minimum at time
\begin{equation}
\bar{t} = -\left( \frac{1}{\lambda'+v} \right) \log \left(  \frac{ K' v }{ \lambda r^{\alpha-d}} \right).
\end{equation}
This implies a minimum:
\begin{align}
h(\bar{t}) &= \left( \frac{K' v}{\lambda' r^{\alpha - d}} \right)^{\frac{\lambda'}{ \lambda' + v}} + \frac{K'}{r^{\alpha - d}}  \left( \frac{K' v} {\lambda' r^{\alpha - d}} \right)^{\frac{-v }{ \lambda' + v }} \nonumber \\
&\leq  c  \left( r^{\alpha - d} \right)^{\frac{-2 \lambda}{ v + 2\lambda}}
\end{align}
for some constant $c$ which depends on $\lambda, v, K$.
Taking the supremum over $f,g$ gives the bound on $T_\sigma (X:Y) $ for $\alpha > d$ in the first line of Eq.~(\ref{eq:bounds-covar-correl}).

For the other two cases, we use the ansatz $t^* = 1+\log(r^\beta)$.
Since the bound in the second line of \cref{eq:LR-bound-cc} scales as
$ \mathcal C(r,t) \propto t^{\alpha-d+1}/r^{\alpha-3d}$
for all $t$, we have
\begin{align}
  h(t^*) &= e^{-\lambda(1+\log(r^\beta))} + K \frac{(1+\log(r^\beta))^{\al-d+1}}{r^{\al-3d}}\nonumber \\
  &= \frac{e^{-\lambda}}{ r^{\lambda\beta}}+ K \frac{(\beta\log(r))^{\al-d+1}}{r^{\al-3d}}
  + \O{\frac{\log^{\al-d}(r)}{r^{\al-3d}}}.
\end{align}
We choose $\beta = (\al-3d)/\lambda$, which is positive for $\al>3d$.
This gives the ultimate bound of
\begin{align}
  h(t^*) &=  \frac{e^{-\lambda}+ K  \left( \frac{\alpha -3d}{\lambda} \log r \right)^{\alpha-d+1}}{r^{\alpha - 3d}} + \O{\frac{\log^{\al-d}(r)}{r^{\al-3d}}}\nonumber\\
  &=
   K\left(\frac{\al-3d}{\lambda}\right)^{\al-d+1}\frac{\log^{\al-d+1}(r)}{r^{\al-3d}} \nonumber\\&\quad \qquad + \O{\frac{\log^{\al-d}(r)}{r^{\al-3d}}},
\end{align}
which proves the second line of Eq.~(\ref{eq:bounds-covar-correl}). For the $d=1$ case in the last line of Eq.~(\ref{eq:bounds-covar-correl}), the argument proceeds similarly, but we obtain a slightly better scaling in the logarithmic factor.
\end{proof}

Here we discuss the scaling of the bounds in \cref{eq:bounds-covar-correl}, which is depicted in \cref{fig:light-cone-scalings}.
The effective exponent of the $1/r$-scaling of the bound for $\al > d$ is $ \al' \equiv (\alpha - d)\frac{2\lambda}{ v + 2\lambda}$, as compared to $\tilde \al \equiv \al-3$ for $\al > 3d$ (neglecting terms doubly logarithmic in $r$).
Since $\al'$ decreases as a function of $v$, the former bound becomes looser for larger $v$.
In more detail, if we let $x = \frac v\lambda$, then $\al' < \tilde \al$ for all $\al > \frac{(3x+4)d}{x}$.
In the limit of $x\ra \infty$, $\tilde \al$ is tighter for all $\al > 3d$.
Thus, for large enough $\al$ and $v$, the power-law light-cone bounds [second line in Eq.~(\ref{eq:LR-bound-cc}), which in turn comes from Eq.~(\ref{eq:LR-Minh-constX})]  give asymptotically tighter bounds on the clustering of covariance correlations than the logarithmic light-cone bound [first line in Eq.~(\ref{eq:LR-bound-cc}), which in turn comes from Eq.~(\ref{eq:LR-ZX-open})].

\begin{figure}[h]
\label{fig:model1}
\includegraphics[width=.45\textwidth]{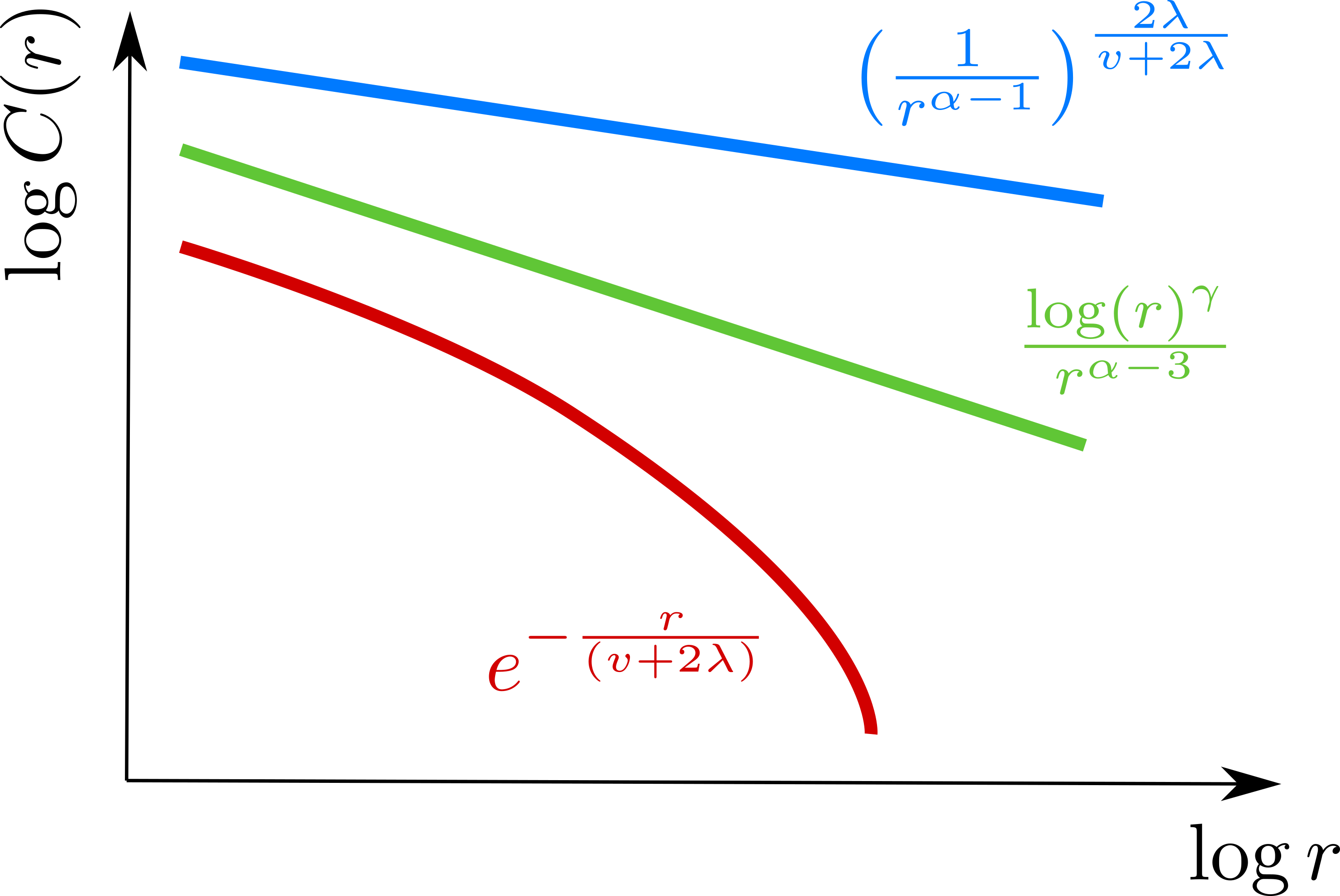}
\caption{A log-log plot of the tails of the bounds on the various connected correlation functions in Theorems 1, 2, and 3 for $d = 1$. We include the exponentially decaying tail from the short-range interactions case (red curve) for comparison.
For the power-law decaying bounds, we have different scaling exponents of the power-law tails for the bound for $\alpha >1$ (blue curve) and the bound for $\alpha > 3$ (green curve).
For a given choice of $x = \frac v\lambda$, the relative positioning of the curves
holds for all $\al > \frac{3x+4}{x}$.
In the limit $v \gg \lambda$, the picture holds for all $\al > 3$.}
  \label{fig:light-cone-scalings}
\end{figure}

\subsection{Stability result and mutual information bound}
In this section, we will  use the aforementioned bounds to constrain steady-state properties of open systems with power-law interactions.
In addition to the newly-derived Lieb-Robinson bounds,  we will  appeal to a ``mixing bound'' which provides an upper bound to how fast an arbitrary initial state will converge to the steady state. The following mixing bound was derived in Ref.~\cite{Kastoryano2013d}, and generalizes the mixing bound of classical Markov chains to quantum semigroups:

\begin{lemma}
  \label{cor:mixing}
    Consider a  primitive Liouvillian $\L$ that has a  full-rank steady state  $\sigma$, and is $\frac12$-reversible. Then an arbitrary initial state $\rho$ will converge to $\sigma$ at a rate bounded by
    \begin{equation}
        \lVert \rho(t) - \sigma \lVert_1 \leq \sqrt{2 \log( \lVert \sigma^{-1} \lVert )} e^{-\beta t},
    \end{equation}
    where $\beta$ is called the log-Sobolev constant associated with $\L$.
\end{lemma}
Intuition can be gained by considering an  ``infinite-temperature'' steady state $\sigma = \mathbb{I} /
d_{H}$ where $ d_{H}$ is the dimension of the Hilbert space.  The mixing bound above states that the coefficient in front of the exponential will  scale as $\sqrt{\log(d_{H})}$, i.e.~it will increase with the dimension of the Hilbert space. This is because the convergence toward the unique steady state from an arbitrary initial state can be slow if the dimension of the Hilbert space is large.

\begin{theorem}[Effect of perturbations on reduced steady-state density matrix]
\label{thm:stability_results}
Let $X,Y$ be two non-overlapping subsets of a $d$-dimensional cubic lattice $\Lambda$.  Let $\L_1 $  be a primitive and $\frac12$-reversible Liouvillian with log-Sobolev constant $\beta$, and let $\L_2$ be a Liouvillian perturbation, acting trivially outside of $X$. Let $\rho$ be the stationary state of $\L_1$, and let $\sigma$ be the stationary state of $\L_1  + \L_2$. Then,
\begin{equation}
\label{eq:perturbation_bounds}
\lVert \rho_Y - \sigma_Y  \lVert_1 \leq
\begin{cases}
    \ds c \log (  \lVert \rho^{-1}  \lVert)^{\frac12} \left(\frac1{r^{\alpha - d}} \right)^{\frac{2 \beta}{ v + 2\beta}},& \al > d,
    \\ \ds c \log (  \lVert \rho^{-1}  \lVert)^{\frac12} \frac{\log(r)^{\al-d+1}}{r^{\al-3d}},& \al > 3d,
    \\ \ds c \log (  \lVert \rho^{-1}  \lVert)^{\frac12} \frac{\log(r)^2}{r^{\al-3}},& \hspace{-2em}\al> 3, d=1,
\end{cases}
\end{equation}
where $c$ is a constant and $r$ is the distance between $X$ and $Y$.
\end{theorem}

The theorem basically says that the effects of local perturbations in the Liouvillian will not be felt significantly by the steady state of the system at sufficiently distant locations.
We prove the theorem by first introducing a time-evolved state to interpolate between the two steady states.
This allows us to use a combination of mixing-time and Lieb-Robinson bounds to restrict the terms in this bound.
Then we apply the same minimization procedure used in \cref{theorem:covariancebound} for the covariance-correlations bound to arrive at the stated bounds in \cref{eq:perturbation_bounds} [each of which follow directly from the three cases in \cref{eq:LR-bound-cc}].
We defer the proof of this result, which is similar to the proof of \cref{theorem:covariancebound}, to \cref{app:proof-stability}.

We now prove a bound on the mutual information in the steady state.  The mutual information between two regions $A,B$ is defined as
\begin{equation}
I_\rho(A:B) = S(\rho_{AB} || \rho_A \otimes \rho_B),
\end{equation}
where $S(\rho || \sigma )= \tr[\rho ( \log \rho - \log \sigma)]$ is the relative entropy. The following theorem holds.

\begin{theorem} [Clustering of mutual information]
\label{thm:mutual_information_clustering}
Let $A,B$ be two non-overlapping subsets of a $d$-dimensional cubic lattice $\Lambda$.  Let $\L $  be a primitive and $\frac12$-reversible Liouvillian with log-Sobolev constant $\beta$. Let $\rho$ be the stationary state of $\L$. Then the mutual information between the two regions $I_\rho(A:B)$ is bounded by
\begin{equation}
I_\rho(A:B)  \leq \begin{cases}
       \ds c \log (  \lVert \rho^{-1}  \lVert)^{\frac32} \left( \frac1{r^{\alpha - d}} \right)^{\frac{2 \beta}{ v + 2\beta}},& \al > d,
    \\ \ds  c \log (  \lVert \rho^{-1}  \lVert)^{\frac32} \frac{\log(r)^{\al-d+1}}{r^{\al-3d}},& \al > 3d,
    \\ \ds c\log (  \lVert \rho^{-1}  \lVert)^{\frac32} \frac{\log(r)^2}{r^{\al-3}},& \hspace{-2em}\al> 3, d=1,
    \end{cases}
\end{equation}
where $c$ is a constant and $r$ is the distance between $A,B$.
\end{theorem}

The significance of this result is that the mutual-information correlations in the steady state of an open long-range system decay as a power-law in the distance between regions.
This bound, which relies on the existence of the log-Sobolev constant, is tighter than the naive bound that would result from simply applying the bound on the covariance correlation in \cref{theorem:covariancebound} to $I_\rho(A:B)$.

\begin{proof}

We define the semi-group $\tilde{\L}$ to be the terms in $\L$ that act entirely within balls of radius $r/2$ centered around $A$ and $B$, and let $\sigma$ be the steady state of  $\tilde{\L}$. Simple manipulations imply:
\begin{align}
I_\rho(A:B) &= - S(\rho_{AB})  + S(\rho_A)  + S(\rho_B) \\
& \leq  - S(\rho_{AB})  - \tr[ \rho_A \log \sigma_A ] -  \tr[ \rho_B \log \sigma_B ] \\
&=  S(\rho_{AB} || \sigma_A \otimes \sigma_B).
\end{align}
where we have used $S(\rho || \sigma) \geq 0$ to obtain the inequality.
The RHS further  satisfies the inequality:
\begin{equation}
S(\rho_{AB} || \sigma_A \otimes \sigma_B) \leq \log( ||\rho_{AB}^{-1} ||) || \rho_{AB} - \sigma_A \otimes \sigma_B ||_1,
\end{equation}
which is a standard result (c.f. Eq.~(36) in Ref.~\cite{Kastoryano2013}).
From here, we can apply the bounds in \cref{thm:stability_results}, using $\L_1 = \td \L $, $\L_2 = \L - \td \L$, $Y = A\cup B$, and $X = \Lambda \setminus Y$.
\end{proof}

\section{Summary and outlook}

In this work, we have proven generalized Lieb-Robinson bounds which constrain the dynamics of open, Markovian systems with power-law interactions and used them to constrain correlations in the steady state.

We comment briefly on the tightness of the bounds derived in this work. Intuitively, one might expect that the presence of dissipation should lead to tighter Lieb-Robinson bounds for open systems than for their closed counterparts, since the presence of decoherence from a bath might  limit the speed of quantum information transfer.
In this work, we have  generalized the proof of Lieb-Robinson bounds from closed system dynamics to Markovian evolution (\textit{a priori}, such bounds need not exist for Markovian dynamics).
However, our bounds only depend on interaction range and the dimension of the lattice.
Any bound that only depends on these two inputs cannot be tighter than the corresponding closed-system Lieb-Robinson bound, since the latter is a special case of former.
As such, the saturating protocols for closed systems \cite{Tran2020hierarchylinearlightcones,Tran2021} can be used to saturate open Lieb-Robinson bounds such as those uncovered in  \cref{lemma:LRbound_truncated}.
In the future, it would be interesting to add another degree of freedom into formulations of open Lieb-Robinson bounds: the dissipative gap. (Some progress has been made in showing that Lieb-Robinson velocities can get tighter in dissipative systems \cite{descamps2013}.)
In principle, it might be possible to derive Lieb-Robinson bounds that reduce to closed-system ones when the dissipative gap is zero, and get tighter in the presence of non-zero dissipation.
Then one can develop protocols that saturate the dissipative-gap-dependent bounds.
Another question in this direction is whether the conditional evolution generated via a non-Hermitian Hamiltonian can also exhibit a dissipative-gap-dependent Lieb-Robinson bound that reduces to the conventional one in the dissipationless limit.

Setting aside the idea of a Lieb-Robinson bound that depends on the dissipative gap, there is still the question of generalizing the best-known closed-system bounds to Markovian evolution. In particular, the recent Lieb-Robinson bounds in Refs.~\cite{kuwaharaStrictlyLinearLight2020} and \cite{Tran2021b} both provide opportunities for generalization to open systems.
Such a result would likely require a modification of the interaction-picture technique first developed in \cite{Foss-Feig2015} and used in both subsequent works to open-system dynamics.
Generalizing these bounds would directly lead to tighter bounds on operator spreading in \cref{lemma:LRbound_truncated} and allow us to prove tighter bounds on correlation clustering in steady states (\cref{theorem:covariancebound} and \cref{thm:mutual_information_clustering})

Another way to probe the tightness of the steady-state correlation bounds derived in this work would be to improve the mixing bounds, which currently require the open system to be in thermal equilibrium.
It would be interesting to derive more general mixing bounds which also apply to systems that are out of thermal equilibrium.

One of the salient applications of Lieb-Robinson bounds is in rigorous proofs on the  stability of  the spectral gap in topologically ordered quantum matter. For example, Ref.~\cite{Bravyi2010} used closed-system Lieb-Robinson bounds to  show that spatially local perturbations will not close energy gaps in the toric code, thus leading to phase stability against arbitrary local noise. Can we use a similar approach to show that local perturbations will not close the dissipative gap in a topologically-ordered open system? A robust qubit steady-state structure would be useful toward the quest of passive quantum error correction \cite{Lidar1998}.

Lieb-Robinson bounds can be used  to prove area-law entanglement scaling in the ground state of one-dimensional systems with local interactions \cite{Hastings2007}. This result helps to rigorously justify the validity of the matrix-product state  ansatz for the ground state of such systems.  For closed systems with power-law interactions, Lieb-Robinson bounds can be used to further extend area-law scaling to certain broad classes of systems \cite{Gong2017}. Do the results presented in this paper have similar implications for area-law scaling of the steady state? This would have direct implications for the matrix-product operator ansatz in modeling open systems.

Finally, the Lieb-Robinson-type bounds we proved apply for the operator, or $\infty$-norm.
However, there exists a hierarchy of Lieb-Robinson-like bounds that have the potential to be tighter for certain information processing tasks such as scrambling and transferring a quantum state of a local subsystem without knowledge of the initial state of the rest of the system.
These bounds can use other norms such as the Frobenius norm defined by $\|O\|_F = \sqrt{\Tr{O^\dag O}}$ \cite{Tran2020hierarchylinearlightcones,Kuwahara2020aPolynomialGrowthOutoftimeorder,Yin2020ScramblingAlltoall,Chen2021Frobenius} or apply to free-particle systems \cite{Guo2019,Tran2020hierarchylinearlightcones}.
It would be interesting to generalize these bounds to open systems as well.

\begin{acknowledgments}
We would like to thank Andrew Lucas, Adam Ehrenberg, and Chris Baldwin for helpful discussions.
AYG is supported by the NSF Graduate Research Fellowship Program under Grant No.~DGE-1840340.
SL was supported by the NIST NRC Research Postdoctoral Associateship Award. AYG, MCT, and AVG acknowledge funding by the DoE ASCR Quantum Testbed Pathfinder program (award No.~DE-SC0019040), DoE Quantum Systems Accelerator, NSF PFCQC program, AFOSR MURI, DoE ASCR Accelerated Research in Quantum Computing program (award No.~DE-SC0020312), ARO MURI, AFOSR,  U.S.~Department of Energy Award No.~DE-SC0019449, and  DARPA SAVaNT ADVENT.
MCT acknowledges additional support from the Department of Energy, Office of Advanced Scientific Computing Research through the QOALAS program (grant 17-020469) and the Quantum Algorithms and Machine Learning grant from NTT, number AGMT DTD 9/24/20.
\end{acknowledgments}

\bibliography{open_cluster}

\begin{thebibliography}{44}
\bibitem [{big()}]{bigO}%
  \BibitemOpen
  \href@noop {} {}\bibinfo {note} {In this paper, we use the following standard
  asymptotic notations: the ``big-$\mathcal{O}$'' notation $f(x) = \O{g(x)}$
  indicates the existence of constants $c$ and $N$ such that $0\le f(x) \le
  cg(x)$ for all $x\ge N$; the ``big-$\Omega$'' notation $f(x) = \W{g(x)}$
  indicates that $g(x) = \O{f(x)}$, and the ``big-$\Theta$'' notation
  $f(x)=\Theta(g(x))$ when both $f(x)=\O {g(x)}$ and $f(x) = \W
  {g(x)}$.}\BibitemShut {Stop}%
\end{thebibliography}

\appendix
\onecolumngrid
\section{Bounds on the error incurred by approximating time-evolved operators by local ones}
\label{app:sim-local-obs}

Here we use the open-system Lieb-Robinson bounds described in \cref{sec:open-LR} of the main text to derive the scalings in \cref{lemma:LRbound_truncated}.
Recall that $\tilde A(t)$ is the evolution of the operator $A$ under the Liouvillian $\td \L$, the restriction of the long-range Liouvillian $\L$ to $\mathcal{B}_r(X)$, the ball of radius $r$ centered on $X$, for time $t$.
We bound the difference between $A(t)$, which is $A$ evolved by the full Liouvillian, and $\tilde A(t)$ as follows:
\begin{align}
  \norm{A(t)-\tilde A(t)}&= \norm{\int_0^t \frac{d}{ds} \left[e^{\L ^\dagger(t-s)} e^{\td \L^\dagger s} A\right]\,\text{d}s} \\
  &= \norm{ \int_0^t e^{\L^\dagger (t-s)} (\L^\dagger-\tilde\L^\dagger)\tilde A(s)\,\text{d}s} \\
  &\le \int_0^t \sum_{j:\text{dist}(j,X)>r} \sum_{i:\text{dist}(i,X)\le r} \norm{\L^\dagger_{ij} \tilde A(s)}\,\text{d}s. \label{eq:clustering-bound}
\end{align}
In order to bound $\norm{\L^\dagger_{ij} \tilde A(s)}$, we turn to the open-system Lieb-Robinson bounds discussed in \cref{sec:open-LR}.
Each line of \cref{eq:LR-bound-cc} and \cref{eq:LR-bound-cc-small-alpha} will correspond to plugging in one of those bounds.
For ease of reference, we reproduce the scalings here:
\begin{equation}
		\label{app:LR-bound-cc}
		\mathcal C(r,t) \propto \begin{cases}
        \ds \frac{e^{\Th{N^{1-\al/d}}t}-1}{\Th{N^{1-\al/d}}},& \al < d,
        \\ \ds \frac{e^{\Th{\log(N)}t}-1}{\Th{\log(N)}},& \al= d.
		\\ \ds \frac{e^{vt}}{r^{\al-d}},& \al > d,
		\\ \ds \frac{t^{\alpha-d+1}}{r^{\alpha-3d}},& \al > 3d,
		\\ \ds \frac{t^2}{r^{\alpha-3}},& \al > 3, d=1.
	\end{cases}
\end{equation}
The calculations will be similar for each bound, so we will only demonstrate the result of inserting the power-law light cone bound from \cref{eq:LR-Minh-constX} into \cref{eq:clustering-bound}:
\begin{align}
  \norm{A(t)-\tilde A(t)} & \le C\norm{A}\int_0^t\text{d}s \sum_{j:\text{dist}(j,X)>r} \sum_{i:d(i,X)\le r} \norm{\L^\dagger_{ij}} \frac{s^{\al-d}}{\text{dist}(i,X)^{\al-2d}} \\
  & \le C\norm{A}\int_0^t\text{d}s \sum_{j:\text{dist}(j,X)>r} \sum_{i:\text{dist}(i,X)\le r} \frac{1}{\text{dist}(i,j)^\al} \frac{s^{\al-d}}{\text{dist}(i,X)^{\al-2d}} \\
  &\le C'\norm{A}\int_0^t\text{d}s \sum_{j:\text{dist}(j,X)>r} \frac{s^{\al-d}}{\text{dist}(j,X)^{\al-2d}} \\
  &\le C''\norm{A}\frac{t^{\al-d+1}}{r^{\al-3d}}.
\end{align}
This yields the expression in the fourth line of \cref{app:LR-bound-cc}.
Performing the same operations for the other bounds gives the other terms in \cref{app:LR-bound-cc}: the first and second lines come from \cref{eq:LR-ZX-open-small-alpha}; the third line comes from \cref{eq:LR-ZX-open}, and the last line comes from \cref{eq:chen-lucas-open-bound}.

\section{Variance  bound for reversible Liouvillians} \label{sec:var-bound}

Here we provide a derivation of the covariance bound used in Eq.~\eqref{eq:cov-bound}.  We show that $s$-reversibility is important for this bound to hold. We define the variance of an observable $f$  in the steady state $\sigma$ as $\text{Var}[f] = \text{Tr}[f^2 \sigma] - \text{Tr}[f \sigma]^2$, which is real and positive. We wish to find a bound for $\text{Var}[f_t]$ for the time-evolved observable $f_t = e^{\L^\dagger t} f$.

The Liouvillian is a non-Hermitian superoperator, which means that each eigenvalue has right and left eigenoperators:
\begin{equation}
\L (r_i) = \lambda_i r_i, \qquad \L^\dagger (l_i) = \lambda_i^* l_i.
\end{equation}
From the structure of the adjoint Liouvillian ($\L^\dagger$), it is clear that $\L^\dagger(\mathbb{I})=0$, where $\mathbb{I}$ is the identity operator. This implies that one of the eigenvalues $\lambda_0$ is zero, and the corresponding right eigenoperator $\sigma$ is called the steady state and satisfies $\L(\sigma)=0$ and $e^{\L t} (\sigma) = \sigma$.  The eigenoperators are ``bi-orthonormal'' via the Hilbert-Schmidt inner product: $\text{Tr}[l_i^\dagger r_j] = \text{Tr}[r_i^\dagger l_j] = \delta_{ij}$.

We define the superoperator $\Gamma_s(f) = ( \sigma^s f \sigma^{1-s} + \sigma^{1-s} f \sigma^s)/2$ where $s \in[0,1]$ and $\sigma$ is a full-rank, Hermitian operator with positive eigenvalues.  We say that a Liouvillian is $s$-reversible for some $s \in[0,1]$ if
$\Gamma_s \L^\dag = \L \Gamma_s$.
By acting both sides on the operator $\mathbb{I}$, we see that $\sigma$ is the steady state, i.e.~that $\L (\sigma)=0$. Imposing reversibility implies that the spectrum must be real because the Liouvillian is pseudo-Hermitian with a positive-definite metric \cite{Mostafazadeh2002}.

The  dynamics preserves Hermiticity of a density matrix, which implies that $\L(f^\dagger) = [\L(f)]^\dagger$, and the same for the adjoint: $\L^\dagger(f^\dagger) = [\L^\dagger(f)]^\dagger$, where $f$ is an arbitrary operator. This implies that (right and left) eigenoperators with real eigenvalues must be Hermitian. For $s$-reversible Liouvillians, the entire spectrum is real, which implies that all eigenoperators are Hermitian.

Ref.~\cite{Kastoryano2013} derives a bound for the time-evolved variance in a $s$-reversible system:
\begin{equation} \label{eq:bound}
\text{Var}[ f_t ] \leq e^{ - 2 \lambda_1 t}  \text{Var}[ f(t=0) ],
\end{equation}
where $\{ - \lambda_i \}$ is the real, non-positive spectrum of $\L$, sorted from smallest to largest magnitude with $\lambda_0 =0, \lambda_1 > 0$.   ($\lambda = \lambda_1$, i.e.~the dissipative gap.) Here we derive this bound using the properties of the eigenoperators of $\L$.

Consider a general Hermitian  operator which we write in terms of left eigenoperators
\begin{equation}
f = \sum_j c_j l_j \Rightarrow f_t = \sum_j c_j e^{ - \lambda_j t} l_j,
\end{equation}
where $c_j$ are real because $f$ is Hermitian.
Noting that  $\text{Tr}[l_j \sigma] = 0$ for $j \neq 0$, we find
\begin{equation}
\text{Var}[ f_t ]  =  \text{Tr} \left[  \left(  \sum_{j \neq 0} c_j e^{ - \lambda_j t} l_j \right)^2  \sigma \right] =      \sum_{j \neq 0} c_j^2 e^{ - 2 \lambda_j t},
\end{equation}
where in the last equality we have used $\text{Tr}[l_i r_j] = \text{Tr}[l_i \Gamma_s(l_j)] = \delta_{ij}$. From this, it is easy to see that $\lambda_{i > 1} \geq \lambda_1 $ implies the bound Eq.~\eqref{eq:bound}.  For the more general case of a complex spectrum, it is not clear how to repeat the derivation above. We therefore find that $s$-reversibility is sufficient for the bound to hold. (It is unclear whether $s$-reversibility is necessary for the bound.)

Given the bound Eq.~\eqref{eq:bound}, one can repeat the steps outlined in Eqs.~[49-55] in Ref.~\cite{Kastoryano2013} to obtain the bound used in Eq.~\eqref{eq:cov-bound} of the main text. For completeness, we include these steps below:
\begin{align}
|\text{Cov}_\sigma (f_t, g_t)| & \leq \sqrt{ \text{Var}(f_t) \text{Var}(g_t)  }  \label{eq:hold} \\
& \leq e^{-2 t \lambda_1} \sqrt{ \text{Var}(f) \text{Var}(g)  }  \label{eq:var1}.
\end{align}
The inequality in \eqref{eq:hold} is due to  Holder's inequality. The variance can be bounded by
\begin{align}
\sqrt{\text{Var}(f)} &= \sqrt{\Tr[ \sigma(f - \Tr[ \sigma f])^2 ] }  \\
& \leq \sqrt{ \lVert (f - \Tr[ \sigma f])^2 \lVert } \\
& \leq \lVert f - \Tr[ \sigma f] \lVert\\
& \leq \lVert f  \lVert + | \Tr[\sigma f]| \\
& \leq 2 \lVert f \lVert  \label{eq:var2}.
\end{align}
Putting together \eqref{eq:var1} and \eqref{eq:var2} leads to the desired bound (where $\lambda = \lambda_1$, i.e.~the dissipative gap):
\begin{equation}
|\text{Cov}_\sigma (f_t, g_t)| \leq 4 \norm{f} \norm{g}  e^{-2 \lambda t},
\end{equation}
which matches \cref{eq:cov-bound} in the main text.

\section{Bound on the difference between two operators evolving separately versus evolving together}
\setcounter{lemma}{1}
\label{app:connectedcorrs}
In this section, we provide the proof of the bound in \cref{lemma:connectedcorrs}. We restate the lemma here for convenience:
\begin{lemma}
  \label{lemma:connectedcorrs2}
    Take two operators $A$ and $B$ supported on single sites $X,Y \in \Lam$  respectively such that $r\coloneqq d(X,Y)$, and let $A(t)=e^{\L^\dag t}A$ and $B(t) = e^{\L^\dag t}B$ be their time-evolution under the Liouvillian superoperator $\L^\dag$.
    We also define $(AB)(t) = e^{\L^\dag t}(AB)$.
    Then the following bound holds:
    \begin{align}
        \|(AB)(t) - A(t)B(t)\| \le K'\|A\|\|B\| \mathcal C(r,t),
    \end{align}
    where $\mathcal C(r,t)$ is given by the Lieb-Robinson-type bound corresponding to the system in question (see \cref{lemma:LRbound_truncated}) and $K'$ is some constant that depends on lattice parameters.
\end{lemma}

\begin{proof}
We define the semi-group $\td \L^\dag$ to be the terms in $\L^\dag$ that act entirely within balls of radius $r/2$ centered around $X$ and $Y$.
Then, let $\tilde A(t)$ be the time-evolved version of $A$ under $\td \L^\dag$ and likewise for $\tilde B(t)$.
By definition, this implies that $\tilde A(t)\tilde B(t) = (\widetilde{AB})(t)$.
We then get
\begin{align}
\label{eq:connectedcorr1}
  \|(AB)(t)-A(t)B(t)\| \le \| (AB)(t) - (\widetilde{AB})(t) \| + \|A(t)B(t)-\tilde A(t) \tilde B(t) \|.
\end{align}
The first term on the RHS of \cref{eq:connectedcorr1} may be bounded by the Lieb-Robinson bound stated in \cref{lemma:LRbound_truncated} (for an operator that is initially supported on two sites instead of one). The second term can be bounded by
\begin{align}
  \|A(t)B(t)-\tilde A(t) \tilde B(t) \| &\le \|A(t)(B(t)-\tilde B(t))\|+\|(A(t)-\tilde A(t))\tilde B(t)\|\\
  &\le \|A\|\|B(t)-\tilde B(t)\|+\|A(t)-\tilde A(t)\|\|B\|,
\end{align}
using $\|A(t)\|\le \|A\|$ and the submultiplicativity of the operator norm.
Using the Lieb-Robinson bound again, we get
\begin{equation}
    \|(AB)(t) - A(t)B(t)\| \le 2K\|A\|\|B\|\mathcal C(r,t),
\end{equation}
which is the same as \cref{eq:ops-evolving-together} in the main text.
\end{proof}

\section{Effect of perturbations on reduced steady-state density matrix}
\label{app:proof-stability}
In this section, we provide the proof of \cref{thm:stability_results}. The argument hews closely to that of Lemma 11 in Ref.~\cite{Kastoryano2013}, but uses the Lieb-Robinson bounds for open long-range systems given in the main text.
\setcounter{theorem}{1}
\begin{theorem}
\label{thm:stability_results}
Let $X,Y$ be two non-overlapping subsets of a $d$-dimensional cubic lattice $\Lambda$.  Let $\L $  be a primitive and $s$-reversible Liouvillian with log-Sobolev constant $\beta$, and let $\mathcal{Q}$ be a local Liouvillian perturbation, acting trivially outside of $X$. Let $\rho$ be the stationary state of $\L$, and let $\sigma$ be the stationary state of $\L  + \mathcal{Q}$. Then,
\begin{equation}
\label{eq:perturbation_bounds_app}
\lVert \rho_Y - \sigma_Y  \lVert_1 \leq \begin{cases}
       c \log (  \lVert \rho^{-1}  \lVert)^{\frac12} \left(\frac1{r^{\alpha - d}} \right)^{\frac{2 \beta}{ v + 2\beta}},& \al > d,
    \\ c \log (  \lVert \rho^{-1}  \lVert)^{\frac12} \frac{\log(r)^{\al-d+1}}{r^{\al-3d}},& \al > 3d,
    \\ c \log (  \lVert \rho^{-1}  \lVert)^{\frac12} \frac{\log(r)^2}{r^{\al-3}},& \al> 3,
    \end{cases}
\end{equation}
where $c$ is some constant, and $r$ is the distance between $X$ and $Y$.
\end{theorem}
\begin{proof}
 We use the following definition of the trace norm:
\begin{equation}
\frac{1}{2} \lVert \rho - \sigma \lVert_1 = \max_{0\leq A \leq \mathbb{I}} \tr[A (\rho-\sigma)],
\end{equation}
for positive semi-definite $A$. This implies
\begin{equation}
\lVert \rho_Y - \sigma_Y \lVert_1 = 2 \tr[ (A_Y \otimes \mathbb{I}_{Y^c}) (\rho - \sigma)],
\end{equation}
where $A_Y = \tr_{Y^c}[ \text{argmax}_{0\leq A \leq \mathbb{I}} \tr[A (\rho-\sigma)]]$.
We use the triangle inequality
\begin{align} \label{eq:triangle}
\tr[ (A_Y \otimes \mathbb{I}_{Y^c}) (\rho - \sigma)] &= \tr[ (A_Y \otimes \mathbb{I}_{Y^c}) \left[(e^{\L t } - e^{(\L +\mathcal{Q})t})(\phi) + (\sigma - e^{(\mL +\mQ) t}(\phi)) + (e^{\mL t}(\phi) - \rho)\right]] \\
&\le \tr[ (A_Y \otimes \mathbb{I}_{Y^c}) (e^{\L t } - e^{(\L +\mathcal{Q})t}) (\phi)] + \frac{1}{2} \lVert \tr_{Y^c}[ \sigma - e^{(\mL +\mQ) t}(\phi) ] \lVert_1
+ \frac{1}{2} \lVert  \tr_{Y^c}[ e^{\mL t}(\phi) - \rho ] \lVert_1,
\end{align}
where $\phi$ is an arbitrary state. Note that we have introduced two time-evolved operators in this step.
We will now use a combination of mixing bounds and Lieb-Robinson bounds to restrict the RHS.
The last term is bounded via the log-Sobolev bound:
\begin{equation}
 \frac{1}{2} \lVert  \tr_{Y^c}[ e^{\mL t}(\phi) - \rho ] \lVert_1  \leq \left( \frac{1}{2} \log (\lVert \rho^{-1} \lVert )  \right)^{\frac12} e^{- \beta t}.
\end{equation}
This is basically an upper bound on how fast an arbitrary initial state must converge towards the steady state.
The second term in Eq.~\eqref{eq:triangle} can be bounded using a combination of Lieb-Robinson bounds and the log-Sobolev bound:
\begin{align}
\frac{1}{2} \lVert \tr_{Y^c}[ \sigma - e^{(\mL +\mQ) t}(\phi) ] \lVert_1  &=   \tr[ A_Y  e^{(\mL +\mQ) t} (  \sigma - \phi) ]   \\
&= \tr[ e^{(\mL^\dagger +\mQ^\dagger) t} (A_Y)   (  \sigma - \phi) ]  \\
&\leq \tr[ (e^{(\mL^\dagger +\mQ^\dagger) t } -  e^{\mL^\dagger t }) (A_Y)   (  \sigma - \phi) ] +  \tr[  e^{\mL^\dagger t } (A_Y)   (  \sigma - \phi) ].
\end{align}
The last term can again be bounded via the log-Sobolev bound:
\begin{align}
\tr[  e^{\mL^\dagger t } (A_Y)   (  \sigma - \phi) ] &\leq \frac{1}{2} \lVert e^{\mL t} (  \sigma - \phi) \lVert_1 \leq \left( 2 \log (\lVert \rho^{-1} \lVert )  \right)^{\frac12} e^{- \beta t}.
\end{align}
The first term can be bounded via the Lieb-Robinson bound:
\begin{align}
\tr[ (e^{(\mL^\dagger +\mQ^\dagger) t } -  e^{\mL^\dagger t }) (A_Y)   (  \sigma - \phi) ]  &\leq   \tr[ (e^{(\mL^\dagger +\mQ^\dagger) t } -  e^{\mL^\dagger t }) (A_Y)  ] \lVert \sigma - \phi \lVert_1 \\
&\leq 2 \tr[ (e^{(\mL^\dagger +\mQ^\dagger) t } -  e^{\mL^\dagger t }) (A_Y)  ] \\
&\leq 2 \tr[ (e^{(\mL^\dagger +\mQ^\dagger) t } -  e^{\mL_{X^c}^\dagger t }) (A_Y)  ] + 2 \tr[ (e^{\mL_{X^c}^\dagger  t } -  e^{\mL^\dagger t }) (A_Y)  ] \\
&\leq K \lVert A_Y \lVert \,\mathcal \mathcal{C}(r, t),
\end{align}
where $\mL_{X^c}$ is the  Liouvillian restricted to terms that do not intersect $X$. $K$ is an arbitrary constant, and $\mathcal \mathcal{C}(r, t)$ is the Lieb-Robinson bound stated in \cref{lemma:LRbound_truncated}.

The first term in Eq.~\eqref{eq:triangle} can be bounded using  the Lieb-Robinson approach above. Gathering all the bounds together leads to
\begin{equation}
\lVert \rho_Y - \sigma_Y  \lVert_1  \leq K_1 \left(\log (\lVert \rho^{-1} \lVert )  \right)^{\frac12} e^{- \beta t} + K_2 \,\mathcal \mathcal{C}(r, t)
\end{equation}
for arbitrary constants $K_1,K_2$. We wish to pick a time $t$ that minimizes the RHS.
We now note that the RHS has the same functional form as the function that we needed to minimize for the covariance correlation bound. Repeating the minimization procedure outlined in Theorem \ref{theorem:covariancebound}, we arrive at the stated bounds in \cref{eq:perturbation_bounds} of the main text.
\end{proof}

\section{Generalization of the Tran \etal~bound to open long-range systems}
\label{sec:minh-bound-proof}
Here we provide the derivation of the open-systems Lieb-Robinson bound in \cref{eq:LR-Minh-constX}.
We use the generalization of the Hastings \& Koma bound to open systems, as described in \cite{Sweke2019}.
Let $K_Y\in \mathbb{L}_Y$ be a Liouvillian with support contained in $Y$ and $\tau(t) \equiv  e^{\L^\dagger t}$ be the backwards time-evolution operator.
The corresponding superoperator bound is
\begin{align}
  \mathcal C(r,t) \equiv \norm{  K_Y(\tau(t) A)} \leq C \|K_Y\|_{\infty} \norm A  \abs{X}\abs{Y}
    \frac{e^{vt}-1}{r^{\alpha}},
    \label{app:LR-HK-open}
\end{align}
If the supports of operators $K_Y$ and $A$ are not constant, then summing \cref{app:LR-HK-open} over the sites in those supports gives a bound of
\begin{align}
  \mathcal C(r,t) \le \|K_Y\|_{\infty} \norm A
    \phi(Y)\frac{e^{vt}}{r^{\alpha-d-1}},
    \label{eq:LR-ZX-open-many-site}
\end{align}
 where $\phi(Y)$ denotes the boundary of $Y$. For simplicity, we will later write this bound in the form
\begin{align}
  C(r,t)\le \|K_Y\|_{\infty}\norm{A}\phi(Y)f(r,t).
  \label{eq:eq:LR-HK-open}
\end{align}

To derive the open-systems Lieb-Robinson bound in \cref{eq:LR-Minh-constX}, we follow the proof in Tran \etal~\cite{Tran2019b}.
We first divide up the time interval $[0,t]$ into $M$ timesteps of size $\Delta t \equiv t/M$ and let $t_i = it/M$ for $i=0,\dots,M$.
For brevity, we denote by $\tau_i \equiv \tau(t_{M-i},t_{M-i+1})$ the time-evolution operator from time $t_{M-i}$ to $t_{M-i+1}$.
We can decompose the evolution of $A$ by $\tau(t)$ into $M$ timesteps:
\begin{align}
  \tau(t)A = \tau_M \tau_{M-1}\dots \tau_1 A.
\end{align}
We then approximate the evolution by $\tau_1$ by a truncated operator $A_1$ such that
\begin{align}
  \norm{\tau_1 A - A_1} = \eps_1,
\end{align}
where $A_1$ is supported on sites at most a distance $\ell$ from the support of $A$.
We repeat the above approximation for the other time intervals to get
\begin{align}
  &\norm{\tau_2 A_1 - A_2} = \eps_2,\\
  &\norm{\tau_3 A_2 - A_3} = \eps_3,\\
   &\dots \nonumber\\
  &\norm{\tau_M A_{M-1} - A_M} = \eps_M.
\end{align}
At the end of this process, we have approximated $\tau(t)A$ by an operator $A_M$ supported on sites located a distance of $M\ell$ from the support of $A$.
We bound the error of this approximation using the triangle inequality:
\begin{align}
  \norm{\tau_M\dots\tau_1 A - A_M} \le \eps_1 + \dots + \eps_M.
\end{align}
By choosing $M\ell$ slightly less than $r$, we guarantee that the support of $A_M$ does not overlap with $X$, which implies that $K_Y(A_M) = 0$ and therefore that the commutator
\begin{equation}
	\mathcal C(r,t) =  \norm{  K_Y(\tau A)} \leq \norm{  K_Y(\tau A-A_M)} + \norm{  K_Y(A_M)} = \norm{  K_Y(\tau A-A_M)}
\end{equation} is at most the error of the approximation: $\eps \equiv \eps_1 + \dots +\eps_M$.
To find a bound on $\eps_1$, we trace out the part of $\tau_1 A$ that lies outside of $\mathcal A_\ell(Y)$, the ball of radius $\ell$ around the support of $A$:
\begin{align}
  A_1 \equiv \frac{1}{\Tr(\mathbb I_{\mathcal A_\ell(Y)^c})} \Tr_{\mathcal A_\ell(Y)^c} (\tau_1 A) \otimes \mathbb I_{\mathcal A_\ell(Y)^c}
     = \int_{\mathcal A_\ell(Y)^c} d\mu(W) W (\tau_1 A) W^\dag,\label{eq:traceint}
\end{align}
where $S^c$ denotes the complement of the set $S$ and the trace is rewritten as an integral over Haar unitaries $W$ supported on $\mathcal A_\ell(Y)^c$, and $\mu(W)$ denotes the Haar measure.

Now the error from approximating $\tau_1 A$ with $A_1$ is given by
\begin{align}
    \eps_1 =\norm{\tau_1 A - A_1}
    &= \norm{\tau_1 A - \int_{\mathcal A_\ell(Y)^c} d\mu(W) W (\tau_1 A) W^\dag}\\
    &= \norm{\int_{\mathcal A_\ell(Y)^c} d\mu(W) \left[\tau_1 A - W (\tau_1 A) W^\dag\right]}\\
    &\leq \int_{\mathcal A_\ell(Y)^c} d \mu (W) \norm{\left[\tau_1 A,W\right]}.
    \label{eq:eps_1-bound}
\end{align}
Plugging this into \cref{eq:eps_1-bound} gives
\begin{align}
  \eps_1 &= \norm{\tau_1 A - A_1} \le  \int_{\mathcal A_\ell(Y)^c} d \mu (W) \norm{A}\phi(Y)f(\ell,\Delta t)
  = \abs{A}\phi(Y) f(\ell,\Delta t),
\end{align}
where $\Delta t = t/M$ is the size of each timestep.
Applying this to all of the errors yields
\begin{equation}
  \eps_j \le \abs{A} \phi(X_j)f(\ell,\Delta t),
\end{equation}
where $X_j$ is the support of $A_j$.
Thus the new bound is
\begin{align}
 \mathcal C(r,t) \le 2\|K_Y\|_{\infty} \eps &\le 2M\|K_Y\|_{\infty}\abs{A} \phi_\text{max}f(\ell,\Delta t)\\
 &= 2 \|K_Y\|_{\infty}\abs{A}\frac{t}{\Delta t}\phi_\text{max}f(\ell,\Delta t),
\end{align}
where $\phi_\text{max} = \max_j \phi(X_j)$, and we replaced $M$ with $t/\Delta t$.
Without loss of generality, we may set $\Delta t = 1$.
Using the form of $f(r,t)$ given in \cref{eq:LR-ZX-open-many-site}, this yields the bound
\begin{align}
   \mathcal C(r,t) &\le C \|K_Y\|_{\infty} \norm A t \phi_{\max} \frac{e^v}{\left(\frac rt\right)^{\al-d-1}} \\
   &\le C \|K_Y\|_{\infty} \norm A \frac{t^{\al-d}}{r^{\al-2d}},
\end{align}
which matches \cref{eq:LR-Minh-constX} in the main text.

\section{Generalization of the Chen \& Lucas bound to open long-range systems}
\label{sec:chen-lucas-bound-proof}
In this section, we provide the proof of the bound in \cref{eq:chen-lucas-open-bound}, which generalizes the closed-system Lieb-Robinson bound from \cite{Chen2019} to open systems.
In the process, we improve the tail of the bound from $1/r$ to $1/r^{\al-2-o(1)}$.
Our goal is to prove that, for an operator $A\in \mathcal{B}(X)$ supported on $X$, for $K_Y\in \mathbb{L}_Y$ a superoperator supported on $Y$, and for backward time-evolution operator $e^{\L^\dagger t}$, we have
\begin{align}
	\norm{  K_Y(e^{\L^\dagger t} A)} \leq C \norm{K_Y}_{\infty} \norm A \frac{t}{r^{\alpha-2}}.
\end{align}
To do that, we use a trivial bound
\begin{align}
    \norm{ K_Y(e^{\L^\dagger t} A)} \le 2\norm{K_Y}_{\infty} \norm{\Py e^{\L^\dagger t}A},
\end{align}
where $\Py$ is the projector onto operators supported on sites at distance $Y$ and beyond.
We will now represent the operator $A$ by its vectorized form $\oket{A}$, so that $\Py$ acting on $A$ can be viewed as a superoperator acting on the vectorized operator: $\Py(A)=\Py\oket{A}$.
Also, from here on out, we will represent $\L^\dagger$ by $\L$ for notational convenience.

The quantity that we wish to bound is $\norm{\Py e^{\L t} \oket{A}}$, which can be expanded in a series
\begin{align}
	\norm{\Py e^{\L t} \oket{A}} &= \sum_{n=0}^\infty \frac{t^n}{n!}\L^n\oket{A} = \sum_{n=0}^{\infty} \frac{t^n}{n!}\sum_{\beta_1,\beta_2,\dots,\beta_n} \L_{\beta_n}\dots\L_{\beta_2} 	\L_{\beta_1}\oket{A},
	\label{eq:seriesexpansion}
\end{align}
where the $\beta_i$ correspond either to single-site terms or two-body couplings, which we will refer to as ``jumps.''

\subsection{More definitions}
We need a few more definitions before we can proceed.
Consider a sequence of jumps $\bm \beta = (\beta_n,\dots,\beta_1)$.
First, we denote by $\nu(\bm \beta)$ the number of jumps in $\bm \beta$ and $\nu_{q}(\bm \beta)$ the number of order-$q$ jumps in $\bm \beta$.
By ``order-$q$'' jumps, we mean jumps that are of length at least $2^{q-1}$ and less than $2^q$.
For example, $\nu_1(\bm \beta)$ is the number of nearest-neighbor jumps in $\bm \beta$.
$\nu_2(\bm\beta)$ counts the number of jumps of length $2,3$.
Given a jump $\beta$, $\dist(\beta,y)$ is the minimum distance from the support of $\beta$ to $y$.
The distance between a sequence of jumps $\bm\beta$ to $y$ is the minimum distance between each jump and $y$.
We also define a number $N_q$ for each $q$ as follows:
\begin{align}
	N_q = \ceil{\frac{\mu}{2^{q\gamma}}\frac{r}{2^q}},\label{eq:Nq1}
\end{align}
where $\gamma\in(0,1)$ is a parameter to be chosen later, and where $\mu<2$ is a constant chosen to be small enough that
\begin{align}
	\sum_{q = 1}^{\infty} (N_q-1) 2^q \leq \mu r \sum_{q=1}^\infty 2^{-q\gamma} < r.
\end{align}

We list the other definitions below (see \cref{fig:def} for a diagram):
\begin{itemize}
\item Given a sequence of jumps $\bm \beta$, we define its $q$-forward subsequence according to \cref{def:q-forward}.
\begin{definition}\label{def:q-forward}
Given a sequence of jumps $\bm \beta = (\beta_n,\dots,\beta_1)$, its $q$-forward subsequence $\bm \lambda^{(q)}$ is constructed as followed:
\begin{itemize}
	\item Set $\bm \lambda^{(q)} = \{\}$ to be an empty sequence and define $\dist(\{\},y) = \dist(x,y)$.
	\item For $j=1,\dots,m$:
	\begin{itemize}
		\item If $\dist(\beta_j,y)<\dist(\bm \lambda^{(q)},y)$ and $\beta_j$ is an order-$q$ jump, add $\beta_j$ to $\bm \lambda^{(q)}$.
	\end{itemize}
\end{itemize}
\end{definition}
We denote by $\mathcal F$ the map from $\bm \beta$ to its set of $q$-forward subsequences $\Lambda = \{\bm \lambda^{(q)}:q = 1,\dots,r\}$. This map is many-to-one.
\item If the $q$-forward subsequence $\bm \lambda^{(q)}$ has at least $N_q$ jumps, we construct the irreducible $q$-forward subsequence $\bm \lambda'^{(q)}$ by taking exactly the first $N_q$ jumps in $\bm \lambda^{(q)}$. Otherwise, we say that there is no irreducible $q$-forward subsequences.
\item
We denote the map from $\Lambda = \{\bm \lambda^{(q)}\}$ to the set of irreducible $q$-forward subsequences $\Lambda'=\{\bm \lambda'^{(q)}\}$ by $\mathcal T$.
Note that $\abs{\Lambda'}$ can be less than $\abs{\Lambda}$ because the length of $\bm \lambda^{(q)}$ may be less than $N_q$ for some $q$.
\item From a set $\Lambda' = \{\bm \lambda'^{(q_1)},\dots,\bm \lambda'^{(q_k)}\}$ of irreducible $q$-forward subsequences, we define $\mathcal I(\Lambda') = \{\bm\beta:\mathcal T(\mathcal F(\bm\beta) ) \supseteq \Lambda'\}$ to be the set of sequences $\bm \beta$ that has $\Lambda'$ in its set of irreducible $q$-forward subsequences.
\end{itemize}

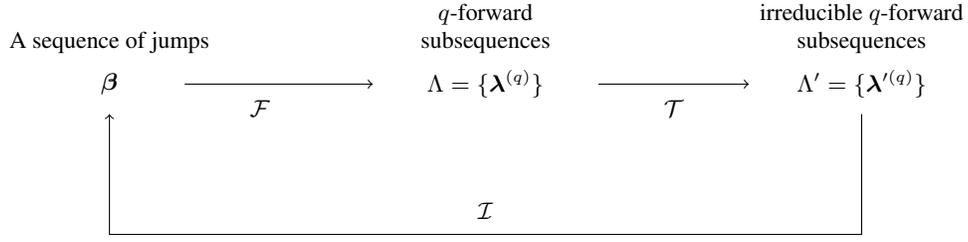
\begin{figure}[h]
\centering
\begin{tikzpicture}
\node[anchor = south] at (0,0.3){
\begin{varwidth}{3cm}
	A sequence of jumps
\end{varwidth}
};
\node[] at (0,0){
	$\bm \beta$
};
\node[] at (0,-0.5){

};
\node[] at (5,0){
	$\Lambda = \{\bm \lambda^{(q)}\}$
};
\node[anchor = south] at (5,0.3){
\begin{varwidth}{3cm}
	$q$-forward subsequences
\end{varwidth}
};
\node[] at (10,0){
	$\Lambda' = \{\bm \lambda'^{(q)}\}$
};
\node[anchor = south] at (10,0.3){
\begin{varwidth}{3cm}
	irreducible $q$-forward subsequences
\end{varwidth}
};
\draw[->] (1,0) -- (3.5,0);
\node[anchor = north] at (2,-0.1){
	$\mathcal{F}$
};
\draw[->] (6.5,0) -- (8.5,0);
\node[anchor = north] at (7.5,-0.1){
	$\mathcal{T}$
};
\draw[->] (10,-0.4) -- (10,-2) -- (0,-2) -- (0,-0.4);
\node[anchor = south] at (5,-1.9){
	$\mathcal{I}$
};
\end{tikzpicture}
\caption{A summary of the definitions regarding sequences and subsequences.}
\label{fig:def}
\end{figure}

\subsection{Proof}
\Cref{lem:exist-long-q} below guarantees that, for each sequence $\bm \beta$ that contributes to \cref{eq:seriesexpansion}, there exists at least one irreducible $q$-forward subsequence $\bm \lambda'^{(q)}$ for some $q$.
\begin{lemma}\label{lem:exist-long-q}
	For each sequence $\bm \beta$,
	if $\Py \L_{\bm \beta}\oket{A}\neq 0$, then there exists at least one $q$-forward subsequence such that $\nu_q(\bm \lambda^{(q)})\geq N_q$.
\end{lemma}
The proof of this lemma is straightforward. If there exists no such $q$, then $\nu_\ell(\bm\lambda^{(q)})\leq N_q-1$ for all $q$.
By the construction of $\bm \lambda$:
\begin{align}
	r \leq \sum_{q = 1}^r \nu_q(\bm\lambda^{(q)}) 2^q
	\leq \sum_{q = 1}^r (N_q-1) 2^q<r,
\end{align}
which is a contradiction.

In the following, we use the notation $\chi_q$ to denote whether $\bm \beta$ has an irreducible $q$-forward subsequence:
\begin{align}
	\chi_q \L_{\bm \beta} \oket{A} =
	\begin{cases}
	\L_{\bm \beta} \oket{A} &\text{if } \exists \bm \lambda'^{(q)} \in \mathcal T(\mathcal F(\bm\beta))),\\
	0 &\text{otherwise.}
	\end{cases}
\end{align}
We can rewrite the series expansion of \cref{eq:seriesexpansion} as
\begin{align}
	\Py e^{\L t}\oket{A} &= \Py \sum_{n=0}^{\infty} \frac{t^n}{n!} \sum_{\bm \beta }\L_{\bm \beta} \oket{A}\\
	&= \Py \left[1-\prod_{q=1}^\infty(1-\chi_q)\right]\sum_{n=0}^{\infty} \frac{t^n}{n!} \sum_{\bm \beta }\L_{\bm \beta} \oket{A},\label{eq:expandchi}
\end{align}
where \cref{lem:exist-long-q} ensures that $1-\prod_{\ell}(1-\chi_\ell) = 1$ for all sequences that contribute to \cref{eq:seriesexpansion}.
Expanding the product over $\ell$, we will get terms of the form
\begin{align}
\mathcal S(q_1,\dots,q_k) &= (-1)^{k+1}\Py \chi_{q_1}\chi_{q_2}\dots\chi_{q_k}\sum_{n=0}^{\infty} \frac{t^n}{n!} \sum_{\bm \beta }\L_{\bm \beta} \oket{A}\\
&=(-1)^{k+1}\Py \sum_{n=0}^{\infty} \frac{t^n}{n!} \sum_{\bm \lambda'^{(q_1)}} \dots\sum_{\bm \lambda'^{(q_k)}}
\sum_{\substack{\bm\beta\in\mathcal I(\{\bm\lambda'^{(q_1)},\dots,\bm\lambda'^{(q_k)}\})\\
\length(\bm\beta) = n
}}
\L_{\bm \beta} \oket{A},
\end{align}
for some distinct integers $q_1,\dots,q_k$.
In the last line, we sum over all possible irreducible $q$-forward subsequences $\bm \lambda^{(q)}$, for $q = q_1,\dots,q_k$, then sum over all sequences $\bm \beta$ which contains $\{\bm\lambda'^{(q_1)},\dots,\bm\lambda'^{(q_k)}\}$ in its set of irreducible $q$-forward subsequences.

We will now upper-bound $\norm{\mathcal S(q_1,\dots,q_k)}$.
First, let $\bm \lambda'$ be a sequence consisting of all jumps in $\bm\lambda'^{(q_1)},\dots,\bm\lambda'^{(q_k)}$ such that the set of irreducible $\ell$-forward subsequences of $\bm \lambda'$ is exactly $\{\bm\lambda'^{(q_1)},\dots,\bm\lambda'^{(q_k)}\}$.
From $\bm\lambda'$, we construct $\bm \beta$:
\begin{align}
	\bm \beta = \left(\beta_{m+1,j_{m+1}},\dots,\beta_{m+1,1}\lambda'_m,\dots,\lambda_2 \beta_{2,j_2},\dots,\beta_{2,1},\lambda'_1,\beta_{1,j_1},\dots,\beta_{1,1}\right),\label{eq:betamunu-power-law}
\end{align}
where $(\lambda'_m,\dots,\lambda'_1) = \bm \lambda'$, $j_1,\dots,j_{m+1}$ are nonnegative integers, $\beta_{i,j} \in \Gamma_i$, and the sets $\Gamma_i$ are constructed recursively for $i=1,\dots,m+1$ as follows:
\begin{itemize}
	\item $
	\Gamma_1 =  \{(x',y'):\dist((x',y'),y)<\dist(x,y) \text{ if } (x',y') \text{ is an order-$q$ jump, where $q = q_1,\dots,q_k$}\}.$
	\item Set $c_q = r$ for all $q = q_1,\dots,q_k$. Each $c_q$ will remember the distance from $y$ to the last length-$q$ jump.
	For the sake of the proof, let $c_q = \infty$ for all other $q$.
	\item For $i = 2$ to $m$:
	\begin{itemize}
	\item $\Gamma_i = \{(x',y'):\dist((x',y'),y)<c_{q(x',y')}.$
	\item Update $c_q = \dist(\lambda'_i,y)$, where $q$ is the order of the jump $\lambda'_i$.
	\end{itemize}
	\item $\Gamma_{m+1} = \{(x',y')\}$ is the set of all possible jumps.
\end{itemize}
The point of this construction is that each sequence $\bm \beta$ appears exactly once.
We can then rewrite
\begin{align}
\mathcal S(q_1,\dots,q_k)
&=(-1)^{k+1}\Py \sum_{n=0}^{\infty} \frac{t^n}{n!} \sum_{\bm \lambda'^{(q_1)}} \dots\sum_{\bm \lambda'^{(q_k)}}
\sum_{\substack{\bm\beta\in\mathcal I(\{\bm\lambda'^{(q_1)},\dots,\bm\lambda'^{(q_k)}\})\\
\length(\bm\beta) = n
}}
\L_{\bm \beta} \oket{A},\\
&=(-1)^{k+1}\Py  \sum_{\bm \lambda'^{(q_1)}} \dots\sum_{\bm \lambda'^{(q_k)}}
\sum_{\bm\lambda'}
\sum_{j_{m+1}=0}^\infty
\dots
\sum_{j_1=0}^\infty
\frac{t^{m+\sum_{l=1}^{m}{j_l}}}{(m+\sum_{l=1}^{m}{j_l})!}
\L_{\Gamma_{m+1}}^{j_{m+1}}\L_{\lambda_{m+1}}\dots\L_{\lambda_{1}}\L_{\Gamma_{1}}^{j_1} \oket{A},\\
&=(-1)^{k+1}\Py  \sum_{\bm \lambda'^{(q_1)}} \dots\sum_{\bm \lambda'^{(q_k)}}
\sum_{\bm\lambda'}
\int_{\Delta^m(t)}dt_1\dots dt_m
e^{\L_{\Gamma_{m+1}}^{j_{m+1}}(t-t_m)}\L_{\lambda_{m+1}}\dots\L_{\lambda_{1}}e^{\L_{\Gamma_{1}}^{j_1}t_1} \oket{A},
\end{align}
where $\Delta^m(t)$ is the simplex defined by $0\leq t_1\leq\dots \leq t_m\leq t$.
Now, we use the triangle inequality:
\begin{align}
\norm{\mathcal S(q_1,\dots,q_k)}&\leq
\frac{3}{2}  \sum_{\bm \lambda'^{(q_1)}} \dots\sum_{\bm \lambda'^{(q_k)}}
\sum_{\bm\lambda'}
\frac{t^m}{m!}
\frac{1}{q_1^{\alpha N_{q_1}}} \dots \frac{1}{q_k^{\alpha N_{q_k}}}\\
&\leq\frac{3}{2}  \binom{r2^q}{N_{q_1}}\dots\binom{r2^q}{N_{q_k}}
\binom{m}{N_{q_1},\dots,N_{q_k}}
\frac{t^m}{m!}
\frac{1}{2^{\alpha q_1 N_{q_1}}} \dots \frac{1}{2^{\alpha q_k N_{q_k}}}\\
&=\frac{3}{2} \prod_{i=1,\dots,k} \left[\binom{r2^q_i}{N_{q_i}}\frac{t^{N_{q_i}}}{N_{q_i}!}
\frac{1}{2^{\alpha q_i N_{q_i}}}\right],
\end{align}
where in the last two lines we use the fact that $m = N_{q_1}+\dots+N_{q_k}$.
Plugging this bound into \cref{eq:expandchi}, we have
\begin{align}
	\norm{\Py e^{\L t}\oket{A}} \leq -1 + \prod_{q} \left[1+\frac{3}{2}\binom{r 2^q}{N_{q}}\frac{t^{N_{q}}}{N_{q}!}
\frac{1}{q^{\alpha N_{q}}}\right].
\end{align}
Now we use $1+x \leq e^x$ to bound
\begin{align}
	\norm{\Py e^{\L t}\oket{A}} \leq -1 + \exp\left[\frac{3}{2}\sum_{q} \binom{r 2^q}{N_{q}}\frac{t^{N_{q}}}{N_{q}!}
	\frac{1}{q^{\alpha N_{q}}}\right].\label{eq:APPF_sumoverq}
\end{align}
Let $q_*$ be the largest integer such that $2^{q_*(\gamma+1)}\leq (\mu r)^{1-\gamma}$.
Note that $\mu r/2^{q(\gamma+1)}>1$ for all $q\leq q_*$.
We divide the sum in \cref{eq:APPF_sumoverq} into two parts:
\begin{align}
	\sum_{q} \binom{r2^q}{N_{q}}\frac{t^{N_{q}}}{N_{q}!}
	\frac{1}{2^{\alpha q N_{q}}}
	&
	\leq \underbrace{\sum_{q=1}^{q_*-1} \binom{r2^q}{N_{q}}\frac{t^{N_{q}}}{N_{q}!}
	\frac{1}{2^{\alpha q N_{q}}}}_{=S_1}
	+\underbrace{\sum_{q=q_*}^{r}\frac{rt}{2^{(\alpha-1) q}}}_{=S_2}.\label{eq:twosum}
\end{align}
First, we estimate $S_2$:
\begin{align}
S_2 \leq \frac{1}{1-2^{-\alpha}}\frac{rt}{2^{q_*(\alpha-1)}}
\leq \underbrace{\frac{1}{1-2^{-\alpha}} \mu^{(1-\alpha)/(\gamma+1)} }_{=c_3}\frac{t}{r^{\frac{\alpha-1}{\gamma+1}-1}}.
\end{align}
Next, we estimate $S_1$.
Note that $N_q\geq \frac{\mu r}{2^{q(\gamma+1)}}$ for all $q$:
\begin{align}
S_1
&\leq \sum_{q=1}^{q_*-1}\left(\frac{e^2rt}{N^2_q 2^{(\alpha-1) q}}\right)^{N_q}\\
&\leq\sum_{q=1}^{q_*}\left(\frac{e^2t}{\mu^2 r}2^{q(2\gamma+3-\alpha)}\right)
^{N_q},\label{eq:boundonS1}
\end{align}
where we have used the Stirling's approximation $x!>x^xe^{-x}$.
When $q\rightarrow 1$, $N_q \propto r$.
The corresponding term in $S_1$ decays with $r$ at least exponentially as $(t/r)^r$.
On the other hand, when $q\rightarrow q_*$, $N_q \rightarrow 1$ and the corresponding term in $S_1$ is instead suppressed by $2^{q(2\gamma+3-\alpha)}$ for all $\alpha>3+2\gamma$.
This limit analysis suggests that we should use two different bounds on $S_1$ for small $q$ and large $q$.
For that, we define
\begin{align}
	q_0 \equiv \floor{\frac{1}{1+\gamma}\log_2(\mu r^\kappa)} \le \frac{1}{1+\gamma}\log_2(\mu r^\kappa)
\end{align}
and divide up $S_1$ into two sums over $q\le q_0$ and $q_0 < q \le q_*$:
\begin{align}
	S_1 \leq \underbrace{\sum_{q=1}^{q_0-1}\left(\frac{e^2t}{\mu^2 r}2^{q(2\gamma+3-\alpha)}\right)^{N_q}}_{=S_{1a}} + \underbrace{\sum_{q=q_0}^{q_*}\left(\frac{e^2t}{\mu^2 r}2^{q(2\gamma+3-\alpha)}\right)^{N_q}}_{=S_{1b}}.
\end{align}
First, we take the sum over $q \le q_0$. We assume that $\alpha > 2\gamma+3$ and $t\le \mu^2r/e^2$, so that the inner summand satisfies
\begin{align}
	\left(\frac{e^2t}{\mu^2 r}2^{q(2\gamma+3-\alpha)}\right) \le 1
\end{align}
for all $q\le q_0$. Because $N_q$ decreases with $q$, we upper bound
\begin{align}
	S_{1a}
	= \sum_{q=1}^{q_0}\left(\frac{e^2t}{\mu^2 r}2^{q(2\gamma+3-\alpha)}\right)^{N_q}
	&\le \left(\frac{e^2t}{\mu^2 r}\right)^{N_{q_0}} \sum_{q=1}^{q_0}2^{q(2\gamma+3-\alpha)N_q } \\
		&\lesssim \left(\frac{e^2t}{\mu^2 r}\right)^\frac{\mu r}{2^{q_0(\gamma+1)}}\\
		&\le \left(\frac{e^2t}{\mu^2 r}\right)^{r^{1-\kappa}}\\
		&\le \frac t r e^{-r^{1-\kappa}},
\end{align}
where in the last line we further assume $t\le \mu^2r/e^2$.
This gives the sum over $q \le q_0$ in the term $S_1$.
To bound the sum over $q_0 < q \le q_*$, we note that $N_{q-1}\geq N_{q} +1$ for all $q<q_*$.
To prove this, suppose $N_{q-1} = N_q$.
That means
\begin{align}
	&\frac{\mu r}{2^{(q-1)(\gamma+1)}}<N_{q-1} = N_{q} \leq \frac{\mu r}{2^{q(\gamma+1)}} +1\\
	\Leftrightarrow & 1 > (2^{\gamma+1}-1) \frac{\mu r}{2^{q(\gamma+1)}} > \frac{\mu r}{2^{q(\gamma+1)}},
\end{align}
which contradicts with $\mu r/2^{q(\gamma+1)}> 1$ for all $q< q_*$.
Therefore, $N_{q-1}\geq N_q + 1$ for all $q<q_*$.
Since $N_{q_*} = 1$, it follows that $N_{q_*-n}\geq n+1>n$ for all $n\geq 1$.
We make the substitution $n= q_* - q$ to obtain
\begin{align}
S_{1b} &= \sum_{q=q_0}^{q_*}\left(\frac{e^2t}{\mu^2 r}2^{q(2\gamma+3-\alpha)}\right)^{N_q} \leq \sum_{n=1}^{q_*-q_0}\left(\frac{e^2t}{\mu^2 r}2^{(q_*-n)(2\gamma+3-\alpha)}\right)^{n},\\
\end{align}
again assuming that $\alpha>3+2\gamma$ and $e^2 t/(\mu^2 r)<1$.
Now, using the fact that $q_*-n \ge q_0$, we have
\begin{align}
	2^{(q_*-n)(2\gamma+3-\alpha)}\le 2^{q_0(2\gamma+3-\alpha)} \leq r^{\kappa(2\gamma+3-\alpha)}.
\end{align}
Plugging this into the sum yields
\begin{align}
\sum_{n=1}^{q_*-q_0}\left(\frac{e^2t}{\mu^2 r}2^{(q_*-n)(2\gamma+3-\alpha)}\right)^{n}
&\leq\sum_{n=1}^{q_*-q_0}\left(\frac{e^2t}{\mu^2 r}r^{\kappa(2\gamma+3-\alpha)}\right)^{n}\\
&= \frac{e^2 t}{\mu^2 r^{1-\kappa(2\gamma+3-\alpha)}}\sum_{n=0}^{q_*-q_0-1}\left(\frac{e^2t}{\mu^2 r^{1-\kappa(2\gamma+3-\alpha)}}\right)^{n} \\
&\leq \frac{e^2 t}{\mu^2 r^{1-\kappa(2\gamma+3-\alpha)}} \frac{1}{1 - \frac{e^2 t}{\mu^2 r^{1-\kappa(2\gamma+3-\alpha)}}}\\
&\leq \underbrace{2\frac{e^2}{\mu^2}}_{=c_2} \frac{t}{r^{1-\kappa(2\gamma+3-\alpha)}},
\end{align}
assuming that $\frac{e^2 t}{\mu^2 r^{1-\kappa(2\gamma+3-\alpha)}}\leq \frac 12$.
Combining everything, we have
\begin{align}
S_1 + S_2 \leq c_1 \left(\frac t r e^{-r^{1-\kappa}}\right) + c_2 \frac t{r^{1-\kappa(2\gamma+3-\alpha)}}+ c_3 \frac{t}{r^{\frac{\alpha-1}{1+\gamma}-1}}.\label{eq:two-terrm-after-sum}
\end{align}
We make the simplification that $\kappa = 1-\gamma$, so that
\begin{align}
	1-\kappa(2\gamma+3-\alpha) = 1-(1-\gamma)(\alpha-3-2\gamma) = \alpha-2\underbrace{-2\gamma-\gamma\alpha+3\gamma+2\gamma^2}_{=o(1)}. \label{eq:sdasdhnks}
\end{align}
In addition, for all $\gamma>0$, there exists a constant $c_\gamma$ that may depend on $\alpha$ such that
\begin{align}
	e^{-r^{\gamma}} \leq c_\gamma \frac{1}{r^{\alpha-3}}
\end{align}
for all $r>0$.
Therefore,
\begin{align}
	\frac{t}{r} e^{-r^{\gamma}} \leq c_\gamma \frac{t}{r^{\alpha-2}}.\label{eq:dskads}
\end{align}
Substituting \cref{eq:dskads,eq:sdasdhnks} into \cref{eq:two-terrm-after-sum} and letting $c = c_1 c_\gamma + c_2+c_3$, we have the desired bound:
\begin{align}
	\norm{\Py e^{\L t} \oket{A}} \le c \frac{t}{r^{\alpha-2-o(1)}},
\end{align}
which is exactly \cref{eq:chen-lucas-open-bound} in the main text.

\end{document}